\documentclass[preprint,12pt]{elsarticle}

\usepackage{geometry}
\usepackage{amsmath,amssymb,amsfonts,mathrsfs}
\usepackage{bm}
\usepackage{color}
\usepackage{graphicx}
\usepackage{float}
\usepackage{subfigure}
\usepackage{multirow}
\usepackage{cases}
\usepackage{extarrows}
\usepackage{subeqnarray}
\usepackage{url}
\usepackage{rotating}

\newcommand{\BlackBox}{\rule{1.5ex}{1.5ex}}  
\newenvironment{proof}{\par\noindent{\bf Proof\ }}{\hfill\BlackBox\\[2mm]}

\newtheorem{theorem}{Theorem}
\newtheorem{lemma}[theorem]{Lemma}
\newtheorem{proposition}[theorem]{Proposition}

\newtheorem{corollary}[theorem]{Corollary}

\DeclareMathOperator*{\argmin}{arg\,min}

\DeclareMathOperator*{\argmax}{arg\,max}
\DeclareMathOperator*{\sign}{sign}

\def\QEDopen{{\setlength{\fboxsep}{0pt}\setlength{\fboxrule}{0.2pt}\fbox{\rule[0pt]{0pt}{1.3ex}\rule[0pt]{1.3ex}{0pt}}}}
\def\QED{\QEDopen}

\journal{arXiv}

\begin{document}
\def \bbR {\mathbb R}
\def \bbN {\mathbb N}
\def \bbZ {\mathbb Z}

\def \mF {\mathcal{F}}
\def \mG {\mathcal{G}}
\def \mI {\mathcal{I}}
\def \mL {\mathcal{L}}
\def \mM {\mathcal{M}}
\def \mN {\mathcal{N}}
\def \mO {\mathcal{O}}
\def \mP {\mathcal{P}}
\def \mQ {\mathcal{Q}}
\def \mS {\mathcal{S}}
\def \mT {\mathcal{T}}

\def \ba {\bm{a}}
\def \bb {\bm{b}}
\def \bc {\bm{c}}
\def \bd {\bm{d}}
\def \bp {\bm{p}}
\def \bq {\bm{q}}
\def \bx {\bm{x}}
\def \by {\bm{y}}
\def \bz {\bm{z}}
\def \bw {\bm{w}}
\def \bu {\bm{u}}
\def \bv {\bm{v}}
\def \br {\bm{r}}
\def \bs {\bm{s}}
\def \bR {\bm{R}}
\def \bS {\bm{S}}
\def \bI {\bm{I}}
\def \bA {\bm{A}}
\def \bB {\bm{B}}
\def \bC {\bm{C}}
\def \bD {\bm{D}}
\def \bE {\bm{E}}
\def \bF {\bm{F}}
\def \bG {\bm{G}}
\def \bH {\bm{H}}
\def \bP {\bm{P}}
\def \bQ {\bm{Q}}
\def \bR {\bm{R}}
\def \bW {\bm{W}}
\def \bX {\bm{X}}

\def \tv {\tilde{v}}
\def \tbx {\tilde{\bx}}
\def \tbv {\tilde{\bv}}
\def \tby {\tilde{\by}}
\def \tbR {\tilde{\bR}}
\def \tbI {\tilde{\bI}}

\def \bxkk {\bx^{k+1}}
\def \bxk {\bx^{k}}
\def \bykk {\by^{k+1}}
\def \byk {\by^{k}}
\def \bwkk {\bw^{k+1}}
\def \bwk {\bw^{k}}
\def \bvkk {\bv^{k+1}}
\def \bvk {\bv^{k}}
\def \tbxkk {\tilde{\bx}^{k+1}}
\def \tbxk {\tilde{\bx}^{k}}
\def \tbykk {\tilde{\by}^{k+1}}
\def \tbyk {\tilde{\by}^{k}}
\def \tbwkk {\tilde{\bw}^{k+1}}
\def \tbwk {\tilde{\bw}^{k}}
\def \tbvkk {\tilde{\bv}^{k+1}}
\def \tbvk {\tilde{\bv}^{k}}

\def \bbRn {\bbR^{n}}
\def \bbRm {\bbR^{m}}
\def \bbRN {\bbR^{N}}
\def \bbRM {\bbR^{M}}
\def \bbNM {\bbN_M}
\def \bbNk {\bbN_k}
\def \bbNn {\bbN_n}
\def \bbNm {\bbN_m}

\def \st {\text{s.\ t.}}
\def \prox {\mathrm{prox}}
\def \Fix {\mathrm{Fix}}

\newcommand\leqs{\leqslant}
\newcommand\geqs{\geqslant}

\begin{frontmatter}

\title{A Krasnoselskii-Mann Proximity Algorithm for Markowitz Portfolios with Adaptive Expected Return Level}
\author[mathjnu]{Yizun Lin}
\ead{linyizun@jnu.edu.cn}

\author[mathjnu]{Yongxin He}
\ead{wheyongxin@163.com}

\author[mathjnu]{Zhao-Rong Lai\corref{mycorrespondingauthor}}
\cortext[mycorrespondingauthor]{Corresponding author}
\ead{laizhr@jnu.edu.cn}

\address[mathjnu]{Department of Mathematics, College of Information Science and Technology, Jinan University,\\ Guangzhou 510632, China}

\begin{abstract}
Markowitz's criterion aims to balance expected return and risk when optimizing the portfolio. The expected return level is usually fixed according to the risk appetite of an investor, then the risk is minimized at this fixed return level. However, the investor may not know which return level is suitable for her/him and the current financial circumstance. It motivates us to find a novel approach that adaptively optimizes this return level and the portfolio at the same time. It not only relieves the trouble of deciding the return level during an investment but also gets more adaptive to the ever-changing financial market than a subjective return level. In order to solve the new model, we propose an exact, convergent, and efficient Krasnoselskii-Mann Proximity Algorithm based on the proximity operator and Krasnoselskii-Mann momentum technique. Extensive experiments show that the proposed method achieves significant improvements over state-of-the-art methods in portfolio optimization. This finding may contribute a new perspective on the relationship between return and risk in portfolio optimization.
\end{abstract}


\end{frontmatter}

\section{Introduction}
\label{intro}
Portfolio optimization (PO) with machine learning methods has become a prospective approach in advancing the interdiscipline of financial engineering \cite{olpsjmlr,SSPO,SPOLC,egrmvgap}. Ever since the first proposal of the mean-variance (MV) approach by Markowitz \cite{PS}, his criterion has become the most popular one for many PO models \cite{POMVreview,valueatrisk,convalueatrisk,sparsepo,enetmeanvariance,meancvar}. In brief, the original MV (OMV) model is
\begin{equation}\label{eqn:omv1}
\begin{split}
&\hat{\bw}=\argmin_{\bw\in \mathbb{R}^N}\bw^\top\bm{\Sigma}\bw,\\
&\st\ \bw^\top\bm{1}_N=1,\ \bw^\top\bm{\mu}=\rho,
\end{split}
\end{equation}
where $\bw$ denotes the $N$-dimensional portfolio (with respect to $N$ assets); $\bm{\mu}$ and $\bm{\Sigma}$ denote the expected return and the return covariance of these $N$ assets, respectively. Constraint $\bw^\top\bm{1}_N=1$ ($\bm{1}_N$ denotes the vector of $N$ ones) is the self-financing constraint, which indicates that no additional money can be used and full re-investment is compulsory. Constraint $\bw^\top\bm{\mu}=\rho$ means that the expected portfolio return is fixed at a level of $\rho$. The objective is to minimize the portfolio variance $\bw^\top\bm{\Sigma}\bw$ (considered as the portfolio risk) at this return level.

Based on many theoretical and practical milestone researches in finance, such as the Capital Asset Pricing Model (CAPM, \cite{CAPM}), the mutual fund performance \cite{SHARPratio} and the efficient market theory \cite{efficientmarket}, a higher portfolio return $\bw^\top\bm{\mu}$ accompanies a higher portfolio risk $\bw^\top\bm{\Sigma}\bw$. Thus they are usually treated as a pair, and the corresponding Pareto optimals form the efficient frontier \cite{CAPM} of all the feasible portfolios. In this sense, different individuals may choose different return levels $\rho$ according to their risk appetites. This convention continues in both theoretical and practical portfolio management today.

On the other hand, machine learning methods have been extending the methodology scope of PO. For example, sparsity methods have been employed to increase portfolio concentration. Brodie et al. \cite{sparsepo} impose $\ell^1$-regularization \cite{LASSO,LARS} on the portfolio to make it sparse and stable. Lai et al. \cite{SSPO} adopt the alternating direction method of multipliers (ADMM, \cite{admm}) to solve a short-term sparse PO model. Luo et al. \cite{SSPOl0} find several closed-form solutions for a short-term sparse PO model with $\ell^0$-regularization. Different from constructing a sparse portfolio, Lai et al. \cite{SPOLC} focus on covariance estimation in PO. They construct a covariance estimate in the perspective of operators and operator spaces. The latter $3$ methods are based on the Exponential Growth Rate (EGR) criterion \cite{problemsetting1,uPS1}, which has a different investing philosophy from the MV criterion \cite{egrmvgap}. Therefore, machine learning methods for the MV criterion are still in great demand.

In the perspective of machine learning, we are inspired to investigate whether it is possible to use an adaptive and flexible return level $\rho$ that fits the ever-changing financial market. It also makes sense in finance: the investor may have no idea about what return level $\rho$ is suitable for her/him, or for the current financial market; All he/she wants may be just getting a reasonable return from the market and getting rid of the trouble to choose a subjective $\rho$. Nevertheless, it is nontrivial to optimize $\rho$ and $\bw$ simultaneously, especially to achieve satisfactory investing performance. It motivates us to develop a novel PO model named Markowitz Portfolio with Adaptive Expected Return Level (MPAERL), which can dynamically balance return and risk. Our main contributions can be summarized as follows.
\begin{itemize}
\item[1)] We develop a new PO model with adaptive expected return level, which including $\ell^1$-regularization, equality constraints and inequality constraints.
\item[2)] We propose a convergent and efficient Krasnoselskii-Mann Proximity Algorithm (KMPA)  which based on the proximity operator and the Krasnoselskii-Mann momentum technique to solve this new PO model.
\item[3)] Our proposed KMPA can be directly extended to solve a class of two-term convex optimization models with inequality constraints.
\end{itemize}

The rest of this paper presents the following contents. Section \ref{sec:relatework} introduces some related works in this field. Section \ref{sec:Model} establishes the MPAERL model. In section \ref{sec_KMPA}, we develop an efficient Krasnoselskii-Mann Proximity Algorithm (KMPA) to solve the MPAERL model. We analyze the convergence of the KMPA in section \ref{mainsec}. Section \ref{sec:experiment} conducts extensive experimental results to assess the performance of MPAERL. Section \ref{sec:conclusion} draws conclusions. Last, we provide the proofs of some technical results in the appendices.

\section{Related Works}\label{sec:relatework}
Brodie et al. \cite{sparsepo} propose the Sparse and Stable Markowitz Portfolios (SSMP) formulated in Lasso \cite{LASSO}
\begin{equation}\label{eqn:ssmp}
\begin{split}
&\hat{\bw}=\argmin_{\bw\in\bbRN}\left\{\frac{1}{T}\| \bR\bw - \rho {\bm1}_{T} \|_2^2+ \tau\| \bw \|_1\right\},\\
&\st \ \bw^\top\hat{\bm{\mu}}=\rho, \bw^\top{\bm1}_N=1,
\end{split}
\end{equation}
where $\bR\in \bbR^{T\times N}$ is the sample asset return matrix ($T$ trading times and $N$ assets), $\br^{(t)}$ denotes the $t$-th row of $\bR$ (i.e., the asset returns at time $t$), $\hat{\bm\mu}:=\frac{1}{T} \bR^\top{\bm1}_{T}$ is a column vector of sample mean returns, $\rho\in\bbR$ is a given expected return level, $\tau\geqs0$ is the regularization parameter, $\|\cdot\|_2$ is the $\ell^2$-norm and $\|\cdot\|_1$ is the $\ell^1$-norm. In this model, the portfolio risk is embedded in the quadratic form $\frac{1}{T}\| \bR\bw - \rho {\bm1}_{T} \|^2$, which computes the mean squared error of the sample portfolio return $\br^{(t)}\bw$ fitting the given level $\rho$. Therefore, SSMP tries to obtain a sparse portfolio that minimizes the risk at the return level $\rho$. To see the relationship between (\ref{eqn:ssmp}) and (\ref{eqn:omv1}), one can expand $\frac{1}{T}\| \bR\bw - \rho {\bm1}_{T} \|^2$ as $\frac{1}{T}(\bw^\top\bR^\top\bR \bw- 2\rho {\bm1}_{T}^\top\bR\bw+\rho^2 T)$, which is a quadratic function of $\bw$ with the symmetric matrix $\frac{1}{T}\bR^\top\bR$. On the other hand, the sample estimator $\hat{\bm{\Sigma}}$ for (\ref{eqn:omv1}) is $\frac{1}{T-1}\bR^\top(\bI_T-\frac{1}{T}{\bm1}_{T}{\bm1}_{T}^\top)\bR$, which is a centralized version of $\frac{1}{T}\bR^\top\bR$. By this way, (\ref{eqn:ssmp}) essentially follows the Markowitz's criterion.

Ho et al. \cite{enetmeanvariance} propose a Weighted Elastic Net Penalized Portfolio (WENPP) that replaces the $\ell^1$-regularization with the elastic net regularization \cite{elasticnet}
\begin{equation}\label{eqn:wenpp}
\begin{split}
\hat{\bw}=\argmin_{\bw\in\bbRN}\ &\bigg\{\bw^\top \hat{\bm{\Sigma}}\bw-\bw^\top \hat{\bm{\mu}}+\sum_{i=1}^N \tau_i |w_i|+\sum_{i=1}^N \kappa_i |w_i|^2\bigg\},
\end{split}
\end{equation}
where $\bw:=(w_1,w_2,\ldots,w_N)^\top$, $\{\tau_i\}_{i=1}^N$ and $\{\kappa_i\}_{i=1}^N$ are the mixing parameters for the $\ell^1$ and $\ell^2$ regularization, respectively. This model can be transformed into a Lasso one via variable changes.

Lai et al. \cite{SSPO} propose a Short-term Sparse Portfolio Optimization (SSPO) that minimizes the negative potential portfolio return with $\ell^1$-regularization
\begin{equation*}\label{eqn:sparsemodel2}
\hat{\bw}=\argmin_{\bw\in\bbRN} \left\{\bw^\top\bm{\varphi}+\tau\|\bw\|_1\right\},\quad\st\  \bw^\top{\bm1}_N=1,
\end{equation*}
where $\bm{\varphi}$ denotes the negative potential asset return, and $\tau$ is the regularization parameter.

Luo et al. \cite{SSPOl0} propose that if the portfolio is further constrained in the simplex
\begin{equation}\label{eqn:portvec}
\Delta_N:=\left\{\bw\in \bbR_+^N: \sum_{i=1}^N w_i=1\right\},
\end{equation}
where $\mathbb{R}_+^N$ is the $N$-dimensional nonnegative space, then the SSPO with $\ell^0$-regularization has closed-form solutions based on the following asset selection $\tilde{\mathbb{I}}_{\bm{\varphi}}^{min}$:
\begin{align}
\label{eqn:sparsemodell0}
&\hat{\bw}=\argmin_{\bw\in\bbRN} \left\{\bw^\top\bm{\varphi}+\tau\|\bw\|_0\right\}, \quad\st\: \bw\in \Delta_N,\\
\label{eqn:sparsemodell0-sol}
&\tilde{\mathbb{I}}_{\bm{\varphi}}^{min}:=\left\{i\in\bbN_{N}:\varphi_i\leqs\min_{j\in\bbN_{N}}\varphi_j+\epsilon\right\},
\end{align}
where $\bm{\varphi}:=(\varphi_1,\varphi_2,\ldots,\varphi_N)^\top$, $\bbN_N:=\{1,2,\ldots,N\}$, $\epsilon\geqs 0$ is a slack variable that allows more assets to be selected and takes the regularizing function of $\tau\|\bw\|_0$.

To fill the gap of covariance estimation in PO, Lai et al. \cite{SPOLC} propose a rank-one covariance estimate
\begin{equation*}\label{eqn:ro-estimate}
\hat{\mathbf{\Sigma}}_{\text{RO}}:=\bu_{1}\zeta_1^*\bu_{1}^\top
\end{equation*}
in the principal rank-one tangent space at the price relative matrix $\bX:=\bR+{\bm1}_{T\times N}$, where $\bu_{1}\in\bbR^{N}$ is the principal right eigenvector in the singular value decomposition of $\bX$ and $\zeta_1^*$ is a computed spectral energy. Then they propose a loss control PO scheme (SPOLC)
\begin{equation*}\label{eqn:spolc1}
\hat{\bw}=\argmax_{\bw \in \Delta_N}\left\{(\min_{1\leqs t\leqs T} \mathbf{x}^{(t)}\bw)-\gamma \bw^\top\hat{\bm{\Sigma}}_{\text{RO}}\bw\right\}
\end{equation*}
with $\hat{\bm{\Sigma}}_{\text{RO}}$, where $\mathbf{x}^{(t)}$ is the $t$-th row of $\bX$, $\min_{1\leqs t\leqs T} \mathbf{x}^{(t)}\bw$ represents the worst increasing factor in the considered time span. SPOLC exploits a trade-off between this worst increasing factor and the risk with a parameter $\gamma>0$, and shows robust performance to the downside risk.

To exploit trend representation in PO, Lai et al. \cite{RPRT} propose a Reweighted Price Relative Tracking (RPRT) system which automatically assigns and updates separate weights to the price relative predictions according to each asset's performance.
\begin{equation*}\label{eqn:rprtmodel}
\begin{split}
&\hat{\bw}_{t+1}=\argmax_{\bw \in \Delta_N} (\bw-\hat{\bw}_{t})^\top \bD_{t+1} (\hat{\bm{\varphi}}_{t+1}-\bar{\varphi}_{t+1}{\bm1}_N),\\
&\hspace{-2pt}\st\ (\bw-\hat{\bw}_{t})^\top (\bD_{t+1}^{-1})^2 (\bw-\hat{\bw}_{t})\leqs\frac{(\max\{\epsilon-\hat{\bw}_{t}^\top\hat{\bm{\varphi}}_{t+1},0\})^2}{\|\hat{\bm{\varphi}}_{t+1}-\bar{\varphi}_{t+1}{\bm1}_N\|_2^2},\qquad
\end{split}
\end{equation*}
where $\bD_{t+1}$ reweights the normalized price relative prediction $(\hat{\bm{\varphi}}_{t+1}-\bar{\varphi}_{t+1}{\bm1}_N)$. The constraint controls the generalized Mahalanobis distance between the candidate weight $\bw$ and the current weight $\hat{\bw}_{t}$ with the square inverse adjustment matrix $(\bD_{t+1}^{-1})^2$.

In recent years, researchers have focused not only on single-period portfolio strategies but also on multi-period investment strategies. The essence of multi-period investment strategies lies in the recognition that the investment outcomes of the current period can influence the risk tolerance or expected return level of the subsequent period. Consequently, the introduction of dynamic risk tolerance/expected-return constraint in portfolio selection was considered more valuable. Along this line of thinking, Wang et al. \cite{liu2015multi} investigated the multi-period portfolio optimization problem with dynamic risk and expected return levels within the mean-variance framework. Later, Gong et al. formulated two multi-period portfolio fuzzy optimization models with certain constraints in \cite{gong2022multi}, namely the wealth maximization model with constrained risk (MCFPS(I)) and the risk minimization model with constrained return (MCFPS(II)). Furthermore, a parameter $a$ was introduced to signify investors attitudes (optimistic, pessimistic, or neutral) towards the stock market.

\section{MPAERL Model}\label{sec:Model}
In this section, we propose the Markowitz Portfolio with Adaptive Expected Return Level (MPAERL). In the SSMP model (\ref{eqn:ssmp}), the expected return level $\rho$ is given manually and fixed according to the risk appetite of an investor, then the portfolio risk is minimized at this return level, which forms a return-risk balancing strategy. However, an investor may not know which return level is suitable for her/him. Besides, this fixed return level may not be suitable for the current financial circumstance. These problems motivate us to design an adaptive expected return level scheme and a more flexible return-risk balancing strategy. To be specific, we allow the expected return level $\rho$ change in an interval and optimize it simultaneously with the portfolio $\bw$ as follows:
\begin{equation}\label{mod_MPAERL1}
\begin{split}
&(\hat{\bw},\hat{\rho})=\argmin_{\bw\in\bbRN,\ \rho\in\bbR}\left\{\frac{1}{T}\|\bR\bw-\rho{\bm1}_T\|_2^2+\tau\|\bw\|_1\right\},\\
&\st\ \bw^\top\hat{\bm\mu}=\rho,\ \bw^\top{\bm1_N}=1,\ \rho_1\leqs\rho\leqs\rho_2,
\end{split}
\end{equation}
where $\rho_1,\rho_2\in(0,+\infty)$ are the given lower and upper bounds of $\rho$, respectively. By this way, we give $\rho$ a loose interval to adapt the financial circumstance, and address the relationships between the return, the risk and the portfolio in a unified framework (abbreviated as MPAERL). In this model, investors can easily adjust the lower bound $\rho_1$ and upper bound $\rho_2$ to suit their requirements, without the necessity of tuning the expected return level $\rho$.

Before developing an efficient algorithm to solve model \eqref{mod_MPAERL1}, we rewrite it as a more compact form. To this end, we let $\bI_N$ denote the $N\times N$ identity matrix, ${\bm0}_N$ denote the vector of $N$ zeros, and define
$$
\bv:=\left(\begin{array}{c}
\bw\\
\rho
\end{array}\right),\ \ \tilde{\bR}:=\left(\bR,\ -{\bm1}_T\right),\ \ \tilde{\bI}:=\left(\bI_N,\ {\bm0}_N\right),
$$
$$
\bA:=\left(\begin{array}{cc}
\hat{\bm\mu}^\top & -1\\
{\bm1}_{N}^\top & 0
\end{array}\right),\ \ \bb:=\left(\begin{array}{c}
0\\
1
\end{array}\right),
$$
$$
\bB:=\left(\begin{array}{cc}
{\bm0}_N^\top & 1\\
{\bm0}_N^\top & -1
\end{array}\right),\ \
\bc:=\left(\begin{array}{c}
\rho_1\\
-\rho_2
\end{array}\right).
$$
Then MPAERL can be rewritten as
\begin{equation}\label{mod_MPAERL2}
\begin{split}
&\hat{\bv}=\argmin_{\bv\in\bbR^{N+1}}\left\{\frac{1}{T}\|\tbR\bv\|_2^2+\tau\|\tbI\bv\|_1\right\},\\
&\st\ \bA\bv=\bb,\ \bB\bv\geqs\bc.
\end{split}
\end{equation}
Note that an equality constraint can be equivalently rewritten as two inequality constraints. By further defining
\begin{equation}\label{defDd}
\bD:=\left(\begin{array}{c}
\bA\\
-\bA\\
\bB
\end{array}\right)\ \ \mbox{and}\ \ \bd:=\left(\begin{array}{c}
\bb\\
-\bb\\
\bc
\end{array}\right),
\end{equation}
model \eqref{mod_MPAERL2} then becomes the following two-term optimization model with an inequality constraint:
\begin{equation}\label{mod_MPAERL3}
\hat{\bv}=\argmin_{\bv\in\bbR^{N+1}}\left\{\frac{1}{T}\|\tbR\bv\|_2^2+\tau\|\tbI\bv\|_1\right\},\quad\st\ \bD\bv\geqs\bd.
\end{equation}

\section{Krasnoselskii-Mann Proximity Algorithm}\label{sec_KMPA}
In this section, we develop an efficient Krasnoselskii-Mann proximity algorithm to solve model \eqref{mod_MPAERL3}. Note that $\bD\in\bbR^{6\times(N+1)}$. To simplify the notation and make the derivation more general, we let $m_1:=N+1$, $m_2:=6$, and define
\begin{equation}\label{def_fandg}
f(\bv):=\frac{1}{T}\|\tbR\bv\|_2^2,\ \ g(\bv):=\tau\|\tbI\bv\|_1,\ \ \mbox{for}\ \ \bv\in\bbR^{m_1}.
\end{equation}
We denote by $\Gamma_0(\bbRm)$ the class of all proper lower semicontinuous convex functions from $\bbRm$ to $\bbR\cup\{+\infty\}$. A function $\psi:\bbRm\to[-\infty,+\infty]$ is said to be proper if $-\infty\notin\psi(\bbRm)$ and $\{\bx\in\bbRm|\psi(\bx)<+\infty\}\neq\varnothing$. It is easy to see that $f\in\Gamma_0(\bbR^{m_1})$ and it is differentiable with a Lipschitz continuous gradient, and $g\in\Gamma_0(\bbR^{m_1})$. In fact, model \eqref{mod_MPAERL3} can be characterized as an equivalent fixed-point problem. To this end, we recall the definitions of proximity operator, subdifferential and conjugate function. Let $\psi\in\Gamma_0(\bbRm)$. The proximity operator of $\psi$ at $\bx\in\bbRm$ is defined by
\begin{equation*}\label{def_prox}
\prox_{\psi}(\bx):=\argmin_{\bu\in\bbRm}\left\{\frac{1}{2}\|\bu-\bx\|_2^2+\psi(\bu)\right\}.
\end{equation*}
The subdifferential of $\psi$ at $\bx\in\bbRm$ is defined by
\begin{equation*}\label{def_subdiff}
\partial\psi(\bx):=\{\by\in\bbRm|\psi(\bu)\geqs\psi(\bx)+\langle\by,\bu-\bx\rangle \ \mbox{for all}\ \bu\in\bbRm\},
\end{equation*}
where $\langle\cdot,\cdot\rangle$ is the inner product defined by $\langle\bx,\by\rangle:=\bx^\top\by$ for $\bx,\by\in\bbRm$. The conjugate function of $\psi$ is given by
\begin{equation}\label{def_conj}
\psi^*(\bx):=\sup_{\bu\in\bbRm}\{\langle\bx,\bu\rangle-\psi(\bu)\},\ \ \mbox{for}\ \ \bx\in\bbRm.
\end{equation}
For an operator $\mT:\bbRm\to\bbRm$, $\bx\in\bbRm$ is called a fixed point of $\mT$ if $\bx=\mT\bx$. We denote the set of all fixed points of $\mT$ by $\Fix(\mT)$. Given an initial vector $\bx^0\in\bbRm$, the fixed-point iteration (Picard iteration) of $\mT$ is given by $\bx^{k+1}=\mT\bx^{k}$. We also define the indicator function $\iota_{\bd}:\bbR^{m_2}\to\bbR\cup\{+\infty\}$ with respect to vector $\bd$ by
\begin{equation}\label{def_iotad}
\iota_{\bd}(\bx):=\begin{cases}
0, & if\ \bx\geq\bd,\\
+\infty, & else.
\end{cases}
\end{equation}
Note that $\iota_{\bd}$ is also convex since the set $\{\bx\in\bbR^{m_2}|\ \bx\geqs\bd\}$ is a convex set. Moreover, $\iota_{\bd}\in\Gamma_0(\bbR^{m_2})$.

We shall construct an operator $\mT_{\beta,\eta}:\bbR^{m_1+m_2}\to\bbR^{m_1+m_2}$ such that a solution of model \eqref{mod_MPAERL3} can be identified as a vector consisting of the first $m_1$ components of a fixed point of $\mT_{\beta,\eta}$. To this end, we let
\begin{equation}\label{def_EandP}
\bE:=\left(\begin{array}{cc}
\bI_{m_1} & -\beta\bD^\top\\
\eta\bD & \bI_{m_2}
\end{array}\right), \bP:=\left(\begin{array}{cc}
\beta\bI_{m_1} &\\
& \eta\bI_{m_2}
\end{array}\right),
\end{equation}
where $\beta$ and $\eta$ are introduced to add two degrees of freedom for controlling the averaged nonexpansiveness of the operator, thereby ensuring the convergence of the algorithm to be proposed subsequently. For
$\bz:=\left(\begin{array}{c}
\bv\\
\by
\end{array}\right)$
with $\bv\in\bbR^{m_1}$ and $\by\in\bbR^{m_2}$, we define function $r:\bbR^{m_1+m_2}\to\bbR$ and operator $\mF:\bbR^{m_1+m_2}\to\bbR^{m_1+m_2}$ by
\begin{equation}\label{def_randmF}
r(\bz):=f(\bv)\ \ \mbox{and}\ \ \mF(\bz):=\left(\begin{array}{c}
\prox_{\beta g}(\bv)\\
\prox_{\eta \iota_{\bd}^*}(\by)
\end{array}\right),
\end{equation}
respectively. Then the operator corresponding to the fixed-point characterization of model \eqref{mod_MPAERL3} is given by
\begin{equation}\label{def_tildeT}
\mT_{\beta,\eta}(\bz):=\mF(\bE\bz-\bP\nabla r(\bz)),\ \ \mbox{for}\ \ \bz\in\bbR^{m_1+m_2}.
\end{equation}
To establish this equivalent result, we recall three known facts (Theorem 16.3 of \cite{bauschke2017convex}, Proposition 2.6 of \cite{micchelli2011proximity} and Theorem 23.5 of \cite{rockafellar1970convex}) in the following lemma that indicate the relationships between minimizer, subdifferential, proximity operator and conjugation.

\begin{lemma}\label{lem_proxsubdiff}
Let $\psi\in\Gamma_0(\bbRm)$. Then the following facts hold:
\begin{itemize}
\item[$(i)$] {\bf(Fermat's rule).} $\hat{\bx}$ is a minimizer of $\psi$ if and only if ${\bm0}\in\partial\psi(\hat{\bx})$.
\item[$(ii)$] $\by\in\partial\psi(\bx)$ if and only if $\bx=\prox_{\psi}(\bx+\by)$.
\item[$(iii)$] $\by\in\partial\psi(\bx)$ if and only if $\bx\in\partial\psi^*(\by)$.
\end{itemize}
\end{lemma}

\begin{theorem}\label{thm_FPchar}
Let $\mT_{\beta,\eta}$ be defined by \eqref{def_tildeT}, $\bz:=\left(\begin{array}{c}
\bv\\
\by
\end{array}\right)$ with $\bv\in\bbR^{m_1}$ and $\by\in\bbR^{m_2}$.
\begin{itemize}
\item [$(i)$]If $\bv$ is a solution of model \eqref{mod_MPAERL3}, then for any $\beta,\eta\in(0,+\infty)$, there exists $\by\in\bbR^{m_2}$ such that $\bz\in\Fix(\mT_{\beta,\eta})$.
\item [$(ii)$]If there exist $\beta,\eta\in(0,+\infty)$ such that $\bz\in\Fix(\mT_{\beta,\eta})$, then $\bv$ is a solution of model \eqref{mod_MPAERL3}.
\end{itemize}
\end{theorem}

\begin{proof}
We first prove item $(i)$. According to the definition of $\iota_{\bd}$, model \eqref{mod_MPAERL3} is equivalent to
\begin{equation}\label{mod_threeterm}
\hat{\bv}=\argmin_{\bv\in\bbR^{m_1}}\left\{f(\bv)+g(\bv)+\iota_{\bd}(\bD\bv)\right\}.
\end{equation}
Suppose that $\bv$ is a solution of model \eqref{mod_MPAERL3}, that is, a solution of model \eqref{mod_threeterm}. By Fermat's rule (Fact $(i)$ of Lemma \ref{lem_proxsubdiff}) and the chain rule of the subdifferential, we have that
\begin{equation}\label{Fermateq}
{\bm0}\in\nabla f(\bv)+\partial g(\bv)+\bD^\top\partial\iota_{\bd}(\bD\bv),
\end{equation}
that is,
$$
-\beta\nabla f(\bv)\in\partial\beta g(\bv)+\beta\bD^\top\partial\iota_{\bd}(\bD\bv),\ \ \mbox{for any}\ \ \beta>0.
$$
Thus, there exists $\by\in\bbR^{m_2}$ such that
\begin{equation}\label{FPyeq1}
\by\in\partial\iota_{\bd}(\bD\bv)
\end{equation}
and
\begin{equation}\label{FPxeq1}
-\beta(\nabla f(\bv)+\bD^\top\by)\in\partial\beta g(\bv).
\end{equation}
Employing Fact $(ii)$ of Lemma \ref{lem_proxsubdiff} for \eqref{FPxeq1}, we see that
\begin{equation}\label{FPxeq2}
\bv=\prox_{\beta g}(\bv-\beta(\nabla f(\bv)+\bD^\top\by)).
\end{equation}
In addition, it follows from \eqref{FPyeq1} and Fact $(iii)$ of Lemma \ref{lem_proxsubdiff} that $\eta\bD\bv\in\partial\eta\iota_{\bd}^*(\by)$, which together with Fact $(ii)$ of Lemma \ref{lem_proxsubdiff} implies that
\begin{equation}\label{FPyeq2}
\by=\prox_{\eta\iota_{\bd}^*}(\by+\eta\bD\bv).
\end{equation}
Now $\bz=\mT_{\beta,\eta}(\bz)$ follows from \eqref{FPxeq2}, \eqref{FPyeq2} and the definition of operator $\mT_{\beta,\eta}$ immediately.

We next prove item $(ii)$. Suppose that $\bz$ is a fixed point of $\mT_{\beta,\eta}$ for some $\beta,\eta\in(0,+\infty)$. By the definition of $\mT_{\beta,\eta}$, we know that \eqref{FPxeq2} and \eqref{FPyeq2} hold. Hence, \eqref{FPyeq1} and \eqref{FPxeq1} hold, which yield \eqref{Fermateq}. Then it follows from Fermat's rule that $\bv$ is a solution of model \eqref{mod_threeterm}, that is, a solution of model \eqref{mod_MPAERL3}.
\end{proof}

According to Theorem \ref{thm_FPchar}, to solve model \eqref{mod_MPAERL3}, it suffices to find a fixed point of operator $\mT_{\beta,\eta}$. As shown in \cite{li2016fast}, the direct fixed-point iteration of $\mT_{\beta,\eta}$ may not converge since $\bE$ is expansive. To guarantee the convergence, we can revise operator $\mT_{\beta,\eta}$ by employing the matrix splitting technique and obtain a new operator that has the same fixed points as $\mT_{\beta,\eta}$ \cite{li2016fast,lin2019krasnoselskii}. Specifically, we let
$$
\bG:=\left(\begin{array}{cc}
\bI_{m_1} & -\beta\bD^\top\\
-\eta\bD & \bI_{m_2}
\end{array}\right),
$$
\begin{equation}\label{defmatW}
\bW:=\bP^{-1}\bG=\left(\begin{array}{cc}
\frac{1}{\beta}\bI_{m_1} & -\bD^\top\\
-\bD & \frac{1}{\eta}\bI_{m_2}
\end{array}\right)
\end{equation}
and define operators $\mT_{\bG}:\bbR^{m_1+m_2}\to\bbR^{m_1+m_2}$ and $\mT_{\bW}:\bbR^{m_1+m_2}\to\bbR^{m_1+m_2}$ by
\begin{equation}\label{def_operTG}
\mT_{\bG}:\bz\to\left\{
\begin{array}{c}
\bu:(\bz,\bu)\ \text{satisfies that}\bigr.\\
\bu=\mF((\bE-\bG)\bu+\bG\bz)
\end{array}
\right\},
\end{equation}
\begin{equation}\label{def_operTW}
\mT_{\bW}:=\mT_{\bG}\circ(\mI-\bW^{-1}\nabla r),
\end{equation}
where $\mI$ denote the identity operator. We note that $G$ and $E$ are essentially the same except for the sign of the lower left block. This is to ensure that $E-G$ is a strictly lower block triangular matrix, allowing the implicit fixed-point iteration $\bv^{k+1}=\mF\big((\bE-\bG)\bv^{k+1}+\bG\bv^k\big)$ to have an explicit form. We show in the following proposition that operator $\mT_{\bG}$ is well-defined, and the two fixed point sets $\Fix(\mT_{\bW})$ and $\Fix(\mT_{\beta,\eta})$ are equivalent.

\begin{proposition}\label{prop_TGwelldef}
Let $\mT_{\bG}$ and $\mT_{\bW}$ be defined by \eqref{def_operTG} and \eqref{def_operTW}, respectively. Then the following hold:
\begin{itemize}
\item[$(i)$] For any given $\bz\in\bbR^{m_1+m_2}$, there exists a unique $\bu\in\bbR^{m_1+m_2}$ such that $\mT_{\bG}(\bz)=\bu$.
\item[$(ii)$] $\Fix(\mT_{\bW})=\Fix(\mT_{\beta,\eta})$.
\end{itemize}
\end{proposition}
\begin{proof}
See Appendix \ref{sec:appendice1}.
\end{proof}

Now to solve model \eqref{mod_MPAERL3}, it suffices to find a fixed point of operator $\mT_{\bW}$, which can be obtained by the fixed-point iteration
\begin{equation}\label{eq:FPiterTW}
\bz^{k+1}=\mT_{\bW}(\bz^k)
\end{equation}
with a given initial vector $\bz^0\in\bbR^{m_1+m_2}$. The convergence of this fixed-point iteration shall be illustrated in next section. We then give the explicit form of the fixed-point iteration \eqref{eq:FPiterTW}. For $\bz^{k}=\left(\begin{array}{c}
\bv^{k}\\
\by^{k}
\end{array}\right)$ with $\bv^{k}\in\bbR^{m_1}$ and $\by^{k}\in\bbR^{m_2}$,
\begin{align}
\notag&\bz^{k+1}=\mT_{\bW}(\bz^{k})=\mT_{\bG}\left(\bz^{k}-\bW^{-1}\nabla r(\bz^{k})\right)\\
\notag\Leftrightarrow\ &\bz^{k+1}=\mF\left((\bE-\bG)\bz^{k+1}+\bG\left(\bz^{k}-\bW^{-1}\nabla r(\bz^{k})\right)\right)\\
\label{implTWiter}\Leftrightarrow\ &\bz^{k+1}=\mF\left((\bE-\bG)\bz^{k+1}+\bG\bz^{k}-\bP\nabla r(\bz^{k})\right)\\
\notag\Leftrightarrow\ &\begin{cases}
\bvkk=\prox_{\beta g}\left(\bvk-\beta(\nabla f(\bvk)+\bD^\top\byk)\right)\\
\bykk=\prox_{\eta\iota_{\bd}^*}\left(\byk+\eta\bD(2\bvkk-\bvk)\right)
\end{cases}
\end{align}
We remark that the second equivalence above holds since the definition of $\bW$ in \eqref{defmatW} implies that $\bG\bW^{-1}=\bP$. It is also worth mentioning that the elements of the block matrix $\bE-\bG$ are all zeros except those in the lower left block, which turns the implicit iteration in \eqref{implTWiter} into an explicit one (see the third equivalence).

For the computation of $\prox_{\eta\iota_{\bd}^*}$, we
need the well-known Moreau decomposition \cite{moreau1962fonctions}, which is recalled as a lemma.
\begin{lemma}[Moreau decomposition]\label{lem_moreaudec}
Let $\psi\in\Gamma_0(\bbRm)$. Then for any $\bx\in\bbRm$, $\bx=\prox_{\psi}(\bx)+\prox_{\psi^*}(\bx)$.
\end{lemma}

Define $\psi(\by):=\eta\iota_{\bd}\left(\by/\eta\right)$, $\by\in\bbR^{m_2}$. Then it is easy to verify from the definition of $\prox_{\psi}(\by)$ that
\begin{equation}\label{eqforproxlstar}
\prox_{\psi}(\by)=\eta\cdot\prox_{\frac{1}{\eta}\iota_{\bd}}\left(\frac{1}{\eta}\by\right).
\end{equation}
We next verify that $\psi^*=\eta\iota_{\bd}^*$. By the definition of conjugate function in \eqref{def_conj}, for any $\by\in\bbR^{m_2}$, we have that
\begin{align*}
\psi^*(\by)&=\sup_{\bu\in\bbR^{m_2}}\left\{\langle\by,\bu\rangle-\eta\iota_{\bd}\left(\frac{\bu}{\eta}\right)\right\}\\
&=\eta\sup_{\bu\in\bbR^{m_2}}\left\{\left\langle\by,\frac{\bu}{\eta}\right\rangle-\iota_{\bd}\left(\frac{\bu}{\eta}\right)\right\}\\
&=\eta\sup_{\bx\in\bbR^{m_2}}\left\{\langle\by,\bx\rangle-\iota_{\bd}(\bx)\right\}=\eta\iota_{\bd}^*(\by).
\end{align*}
Then the fact $\psi^*=\eta\iota_{\bd}^*$ together with Lemma \ref{lem_moreaudec} and \eqref{eqforproxlstar} implies that

\begin{align}
\notag\prox_{\eta\iota_{\bd}^*}(\by)&=\prox_{\psi^*}(\by)=\by-\prox_{\psi}(\by)\\
\label{equaprox}&=\eta(\mI-\prox_{\frac{1}{\eta}\iota_{\bd}})\left(\frac{1}{\eta}\by\right),\ \ \mbox{for}\ \ \by\in\bbR^{m_2}.
\end{align}

To implement iteration \eqref{eq:FPiterTW}, we still need the closed forms of $\prox_{\beta g}$ and $\prox_{\frac{1}{\eta}\iota_{\bd}}$, where $g$ and $\iota_{\bd}$ are defined by \eqref{def_fandg} and \eqref{def_iotad}, respectively. By the definition of the indicator function $\iota_{\bd}$, we know that $\iota_{\bd}=\frac{1}{\eta}\iota_{\bd}$, which together with the definition of proximity operator yields that
\begin{equation*}
\prox_{\frac{1}{\eta}\iota_{\bd}}(\by)=\prox_{\iota_{\bd}}(\by)=\max(\by,\bd),
\end{equation*}
where the maxima in the above equation is taken component-wise. In addition, it is easy to verify that for $\bv\in\bbR^{m_1}$,
\begin{align*}
&\prox_{\beta g}(\bv)=\prox_{\beta\tau\|\cdot\|_1\circ\tbI}(\bv)\\
=&\big(\prox_{\beta\tau|\cdot|}(v_1),\prox_{\beta\tau|\cdot|}(v_2),\ldots,\prox_{\beta\tau|\cdot|}(v_{m_1-1}),v_{m_1}\big)^\top,
\end{align*}
where
$$
\prox_{\beta\tau|\cdot|}(x)=\max(|x|-\beta\tau,0)\cdot\sign(x),\ \ \mbox{for}\ \ x\in\bbR
$$
is the soft thresholding operator (see Example 2.3 of \cite{micchelli2011proximity}).

Though the fixed-point iteration of $\mT_{\bW}$ can guarantee the convergence, we may also care about the speed of convergence. To accelerate the convergence speed while preserving the theoretical convergence, the Krasnoselskii-Mann (KM) momentum technique can be utilized. The use of KM momentum scheme obtains a better approximation of the solution by adding the current fixed-point update to the difference between the current fixed-point update and the update from prior iteration. Specifically, the KM iteration for solving model \eqref{mod_MPAERL3} is given by
\begin{equation}\label{KMiterscheme}
\bz^{k+1}=\mT_{\theta_k}\bz^{k}=\mT_{\bW}\bz^{k}+\theta_k\left(\mT_{\bW}\bz^{k}-\bz^{k}\right),
\end{equation}
where
\begin{equation}\label{def_Tthetak}
\mT_{\theta_k}:=(1+\theta_k)\mT_{\bW}-\theta_k\mI,\ \ k\in\bbN,
\end{equation}
$\bbN$ denotes the set of all nonnegative integers. Throughout this paper, the momentum parameter is set to $\theta_k=\frac{\varrho k}{k+\delta}$, where $\delta\in(0,+\infty)$, $\varrho\in(-1,1)$, $k\in\bbN$. According to the explicit form of iteration \eqref{implTWiter} and Equation \eqref{equaprox}, the KM iteration in \eqref{KMiterscheme} can be written as
\begin{equation*}\label{iterscheme1}
\begin{cases}
\tbvkk=\prox_{\beta g}\left(\bvk-\beta(\nabla f(\bvk)+\bD^\top\byk)\right)\\
\tbykk=\eta(\mI-\prox_{\frac{1}{\eta}\iota_{\bd}})\left(\frac{1}{\eta}\byk+\bD(2\tbvkk-\bvk)\right)\\
\theta_k=\frac{\varrho k}{k+\delta}\\
\bvkk=(1+\theta_k)\tbvkk-\theta_k\bvk\\
\bykk=(1+\theta_k)\tbykk-\theta_k\byk
\end{cases}.
\end{equation*}
The last two steps in the above iteration is the KM momentum scheme. We call this iterative scheme Krasnoselskii-Mann Proximity Algorithm (KMPA).

The setting of $\theta_k$ is grounded in reference \cite{lin2019krasnoselskii}, ensuring both convergence and the robustness of convergence. We also remark that the KMPA can be extended to solve portfolio optimization models with non-convex constraints, such as cardinality and bounding constraints, which are common in portfolio optimization to ensure diversification across a specified number of assets and to limit the capital allocated to each asset. However, in the non-convex case, it can only guarantee the convergence to a critical point or locally optimal solution rather than a globally optimal solution.

\section{Convergence Analysis of KMPA}\label{mainsec}
In this section, we analyze the convergence of KMPA. To this end, we recall the definitions of nonexpansiveness, firm nonexpansiveness and averaged nonexpansiveness. Let $\bH\in\bbR^{m\times m}$ be a symmetric positive definite matrix. The weighted norm $\|\cdot\|_{\bH}$ is defined by $\|\bx\|_{\bH}:=\langle\bx,\bH\bx\rangle^\frac{1}{2}$, for $\bx\in\bbRm$. An operator $\mT:\bbRm\to\bbRm$ is called nonexpansive with respect to $\bH$ if $\|\mT\bx-\mT\by\|_{\bH}\leqs\|\bx-\by\|_{\bH}$ for all $\bx,\by\in\bbRm$. If $\|\mT\bx-\mT\by\|_{\bH}^2\leqs\langle\mT\bx-\mT\by,\bx-\by\rangle_{\bH}$ for all $\bx,\by\in\bbRm$, we say that $\mT$ is firmly nonexpansive with respect to $\bH$. If there exists a nonexpansive operator $\mN:\bbRm\to\bbRm$ with respect to $\bH$ and $\alpha\in(0,1)$ such that $\mT=(1-\alpha)\mI+\alpha\mN$, we say that $\mT$ is $\alpha$-averaged nonexpansive with respect to $\bH$. We also recall the following KM theorem \cite{bauschke2017convex,krasnosel1955two,mann1953mean}, which is crucial for the proof of convergence. It is easy to see from this theorem that an averaged nonexpansive operator has a convergent fixed-point iteration.

\begin{theorem}[KM theorem]\label{thm_KM}
Let $\mN:\bbRm\to\bbRm$ be a nonexpansive operator with respect to some symmetric positive definite matrix, such that $\Fix(\mN)\neq\varnothing$. For $\{\alpha_k\}_{k\in\bbN}\subset[0,1]$ and $\bx^0\in\bbRm$, define
$$
\bx^{k+1}:=(1-\alpha_k)\bx^{k}+\alpha_k\mN\bx^{k},\ \ k\in\bbN.
$$
If $\sum_{k=0}^\infty\alpha_k(1-\alpha_k)=+\infty$, then $\{\bx^{k}\}_{k\in\bbN}$ converges to a fixed point of $\mN$.
\end{theorem}

We then show the averaged nonexpansiveness of operator $T_{\bW}$. We denote by $\lambda_{\min}(\bW)$ the minimum eigenvalue of $\bW$, $L\in(0,+\infty)$ the Lipschitz constant of $\nabla f$, that is, $\|\nabla f(\bx)-\nabla f(\by)\|_2\leqs L\|\bx-\by\|_2$ for all $\bx,\by\in\bbR^{m_1}$, and define
\begin{equation}\label{def_zeta}
\zeta:=\frac{2\lambda_{\min}(\bW)}{4\lambda_{\min}(\bW)-L}.
\end{equation}

\begin{proposition}\label{prop1_TWaver}
Let $\bW$ and $\mT_{\bW}$ be defined by \eqref{defmatW} and \eqref{def_operTW}, respectively. If $\lambda_{\min}(\bW)>\frac{L}{2}$, then $\bW$ is symmetric positive definite and $\mT_{\bW}$ is $\zeta$-averaged nonexpansive with respect to $\bW$.
\end{proposition}
\begin{proof}
See Appendix \ref{sec:appendice2}.
\end{proof}

We also recall Lemma 6.2 of \cite{li2015multi} as the following Lemma \ref{lem_pdequiv}.

\begin{lemma}\label{lem_pdequiv}
For symmetric positive definite matrices $\bE_1\in\bbR^{n\times n}$, $\bE_2\in\bbR^{m\times m}$ and an $m\times n$ real matrix $\bC$, let $\bF:=\left(\begin{array}{cc}
\bE_1&\bC^\top\\
\bC&\bE_2
\end{array}\right)$ and $\widetilde{\bC}:=\bE_2^{-\frac{1}{2}}\bC\bE_1^{-\frac{1}{2}}$. Then $\bF$ is positive definite if and only if $\|\widetilde{\bC}\|_2<1$.
\end{lemma}

\begin{corollary}\label{cor:neqlambdamin}
Let $\bW$ be defined by \eqref{defmatW} and $\xi\in(0,+\infty)$, If
$$
\beta\in\left(0,\frac{2\xi}{L}\right)\ \text{and}\ \eta\in\left(0,\frac{2\xi(2\xi-\beta L)}{4\beta\xi^2\|\bD\|_2^2+L(2\xi-\beta L)}\right),
$$
then $\lambda_{\min}(\bW)>\frac{L}{2\xi}$.
\end{corollary}
\begin{proof}
To prove that $\lambda_{\min}(\bW)>\frac{L}{2\xi}$, it suffices to show that
$$
\bW-\frac{L}{2\xi}\bI_{m_1+m_2}=\left(\begin{array}{cc}
\left(\frac{1}{\beta}-\frac{L}{2\xi}\right)\bI_{m_1} & -\bD^\top\\
-\bD & \left(\frac{1}{\eta}-\frac{L}{2\xi}\right)\bI_{m_2}
\end{array}\right)
$$
is positive definite. Since $\beta\in\left(0,\frac{2\xi}{L}\right)$ and $\eta\in\left(0,\frac{2\xi(2\xi-\beta L)}{4\beta\xi^2\|\bD\|_2^2+L(2\xi-\beta L)}\right)$, we have that $\frac{1}{\beta}-\frac{L}{2\xi}>0$ and $\frac{1}{\eta}-\frac{L}{2\xi}>0$. Let $\widetilde{\bD}:=\frac{1}{\sqrt{\left(\frac{1}{\beta}-\frac{L}{2\xi}\right)\left(\frac{1}{\eta}-\frac{L}{2\xi}\right)}}\bD$.
It follows from Lemma \ref{lem_pdequiv} that $\bW-\frac{L}{2\xi}\bI_{m_1+m_2}$ is positive definite if and only if $\|\widetilde{\bD}\|_2<1$, that is,
$$
\left(\frac{1}{\beta}-\frac{L}{2\xi}\right)\left(\frac{1}{\eta}-\frac{L}{2\xi}\right)>\|\bD\|_2^2.
$$
Using the facts $\beta\in\left(0,\frac{2\xi}{L}\right)$ and $\eta\in\left(0,\frac{2\xi(2\xi-\beta L)}{4\beta\xi^2\|\bD\|_2^2+L(2\xi-\beta L)}\right)$ again, we obtain
\begin{align*}
\left(\frac{1}{\beta}-\frac{L}{2\xi}\right)\left(\frac{1}{\eta}-\frac{L}{2\xi}\right)&>\left(\frac{1}{\beta}-\frac{L}{2\xi}\right)\left(\frac{4\beta\xi^2\|\bD\|_2^2+L(2\xi-\beta L)}{2\xi(2\xi-\beta L)}-\frac{L}{2\xi}\right)\\
&=\frac{2\xi-\beta L}{2\beta\xi}\cdot\frac{2\beta\xi\|\bD\|_2^2}{2\xi-\beta L}=\|\bD\|_2^2,
\end{align*}
which completes the proof.
\end{proof}

According to Proposition \ref{prop1_TWaver} and Corollary \ref{cor:neqlambdamin}, we have the following Proposition \ref{prop_TWaver} that provides the ranges of $\beta$ and $\eta$ to guarantee the averaged nonexpansiveness of $\mT_{\bW}$.

\begin{proposition}\label{prop_TWaver}
Let $\bW$, $\mT_{\bW}$ and $\zeta$ be defined by \eqref{defmatW}, \eqref{def_operTW} and \eqref{def_zeta}, respectively. If $\beta\in\left(0,\frac{2}{L}\right)$ and $\eta\in\left(0,\frac{2(2-\beta L)}{4\beta\|\bD\|_2^2+L(2-\beta L)}\right)$, then $\mT_{\bW}$ is $\zeta$-averaged nonexpansive with respect to $\bW$.
\end{proposition}
\begin{proof}
According to Proposition \ref{prop1_TWaver}, to prove the desired result, it suffices to verify that $\lambda_{\min}(\bW)>\frac{L}{2}$, which follows from Corollary \ref{cor:neqlambdamin} with $\xi=1$ immediately.
\end{proof}

It has been shown in Proposition \ref{prop_TWaver} that $\mT_{\bW}$ is $\zeta$-averaged nonexpansive with respect to the symmetric positive definite matrix $\bW$ for appropriate choices of $\beta$ and $\eta$. Then by employing Theorem \ref{thm_KM}, we know that the sequence generated by the fixed-point iteration of $\mT_{\bW}$ converges to a fixed point of $\mT_{\bW}$, that is, a minimizer of model \eqref{mod_MPAERL3}. In addition, we have the following convergence result for KMPA.

\begin{theorem}
Let $\mT_{\bW}$ be defined by \eqref{def_operTW}, $\xi:=1-\max\{\varrho,0\}$, $\bz^{0}\in\bbR^{m_1+m_2}$ be any initial vector, $\{\bz^{k}\}_{k\in\bbN}$ be the sequence generated by \eqref{KMiterscheme} and $\bx^k:=\left(z_1^{k},z_2^{k},\ldots,z_{m_1}^{k}\right)^\top$, $k\in\bbN$. If $\beta\in\left(0,\frac{2\xi}{L}\right)$ and $\eta\in\left(0,\frac{2\xi(2\xi-\beta L)}{4\beta\xi^2\|\bD\|_2^2+L(2\xi-\beta L)}\right)$, then $\{\bx^{k}\}_{k\in\bbN}$ converges to a solution of model \eqref{mod_MPAERL3}.
\end{theorem}
\begin{proof}
Since $\varrho\in(-1,1)$, we have $\xi\in(0,1]$. Corollary \ref{cor:neqlambdamin} yields that $\lambda_{\min}(\bW)>\frac{L}{2\xi}>\frac{L}{2}$. Then it follows from Proposition \ref{prop1_TWaver} that $\bW$ is positive definite and $\mT_{\bW}$ is $\zeta$-averaged nonexpansive with respect to $\bW$, where $\zeta$ is defined by \eqref{def_zeta}. This implies that there exists a nonexpansive operator $\mM$ with respect to $\bW$ such that $\mT_{\bW}=(1-\zeta)\mI+\zeta\mM$. Hence
\begin{align}
\notag \mT_{\theta_k}&=(1+\theta_k)[(1-\zeta)\mI+\zeta\mM]-\theta_k\mI\\
\label{eq_Tthetakaver2}&=[1-(1+\theta_k)\zeta]\mI+(1+\theta_k)\zeta\mM,
\end{align}
for all $k\in\bbN$. To employ Theorem \ref{thm_KM} to prove this theorem, it suffices to verify that $(1+\theta_k)\zeta\in[0,1]$ and
$$
\sum_{k=0}^\infty(1+\theta_k)\zeta[1-(1+\theta_k)\zeta]=+\infty.
$$
Recall that $\theta_k=\frac{\varrho k}{k+\delta}\in(-|\varrho|,|\varrho|)$. This yields that $(1+\theta_k)\zeta>(1-|\varrho|)\zeta>0$. Let $\zeta':=(1+\max\{\varrho,0\})\zeta$. Then $\zeta'>0$ and $(1+\theta_k)\zeta<\zeta'$. In addition, the inequality $\lambda_{\min}(\bW)>\frac{L}{2\xi}$, combined with $\xi=1-\max\{\varrho,0\}$, gives that $\max\{\varrho,0\}<1-\frac{L}{2\lambda_{\min}(\bW)}$, and hence $1+\max\{\varrho,0\}<\frac{4\lambda_{\min}(\bW)-L}{2\lambda_{\min}(\bW)}$. Then we see from the definitions of $\zeta'$ and $\zeta$ that
$$
\zeta'=(1+\max\{\varrho,0\})\frac{2\lambda_{\min}(\bW)}{4\lambda_{\min}(\bW)-L}<1.
$$
Now we conclude that
\begin{equation}\label{neq:rangexiprime}
0<(1-|\varrho|)\zeta<(1+\theta_k)\zeta<\zeta'<1.
\end{equation}
Note that $(1-|\varrho|)\zeta(1-\zeta')$ is a positive constant. Then \eqref{neq:rangexiprime} implies that
$$
\sum_{k=0}^\infty(1+\theta_k)\zeta[1-(1+\theta_k)\zeta]>\sum_{k=0}^\infty(1-|\varrho|)\zeta(1-\zeta')=+\infty.
$$
To employ Theorem \ref{thm_KM} for proving the convergence of $\{\bz^{k}\}_{k\in\bbN}$, we still need the existence of fixed point of $\mT_{\bW}$. This can be achieved as long as model \eqref{mod_MPAERL3} has a solution (see Theorem \ref{thm_FPchar} $(i)$ and Proposition \ref{prop_TGwelldef} $(ii)$). Note that the objective function in model \eqref{mod_MPAERL3} is proper, lower semicontinuous, convex and coercive. The constraint set of this model is closed and convex. Then it follows from Proposition 11.15 of \cite{bauschke2017convex} that model \eqref{mod_MPAERL3} has a solution, and hence $\Fix(\mT_{\bW})\neq\varnothing$.

Now according to Theorem \ref{thm_KM}, we know that $\{\bz^{k}\}_{k\in\bbN}$ converges to a fixed point $\bz^*$ of $\mT_{\bW}$. We also know from Theorem \ref{thm_FPchar} and Item $(ii)$ of Proposition \ref{prop_TGwelldef} that the vector composed by the first $m_1$ components of $\bz^*$ is a solution of model \eqref{mod_MPAERL3}. Therefore, $\{\bx^{k}\}_{k\in\bbN}$ converges to a solution of model \eqref{mod_MPAERL3}.
\end{proof}

To close this section, we summarize the whole MPAERL strategy as the following Algorithm 1.

\hspace{-1.5em}\rule[0em]{16.6cm}{0.1em}\\
\text{{\bf Algorithm 1.} Whole MPAERL strategy}
\hspace{-18.6em}\rule[-1em]{16.6cm}{0.1em}\\

\noindent{\bf Input:} Given the sample asset price relative matrix $\bX\in\bbR^{T\times N}$, the regularization parameter $\tau\geqs0$, the lower bound $\rho_1$ and the upper bound $\rho_2$ of the expected return level, the momentum parameters $\varrho\in(-1,1)$ and $\delta>0$. Set the tolerance $tol=10^{-8}$ and the maximum iteration number $MaxIter=10^4$.\vspace{0.25em}

\hspace{-1.5em}{\bf Initialization:} Compute the sample asset return matrix $\bR=\bX-{\bm1}_{T\times N}$ and the sample mean return vector $\hat{\bm\mu}=\frac{1}{T} \bR^\top{\bm1}_{T}$. Set $\bv^0=\frac{1}{N}{\bm1}_{N+1}$, $v_{N+1}^0=\frac{1}{2}(\rho_1+\rho_2)$, $\by^{0}=\bD\bv^0$; and let $\tilde{\bR}=\left(\bR,\ -{\bm1}_T\right)$, $\bD$ and $\bd$ be given by \eqref{defDd}.\vspace{0.25em}

\hspace{-1.5em}1. Compute the Lipschitz constant $L=\frac{2}{T}\|\tilde{\bR}^\top\tilde{\bR}\|_2$.

\hspace{-1.5em}2. $\xi=1-\max\{\varrho,0\}$, $\beta=\frac{\xi}{L}$, $\eta=\frac{\xi(2\xi-\beta L)}{4\beta\xi^2\|\bD\|_2^2+L(2\xi-\beta L)}$ and $k=0$.\vspace{0.25em}

\hspace{-1.5em}{\bf repeat}

3. $\tbvkk=\prox_{\beta\tau\|\cdot\|_1\circ\tbI}\left(\bvk-\beta\left(\frac{2}{T}\tbR^\top\tbR\bvk+\bD^\top\byk\right)\right)$

4. $\tbykk=\eta(\mI-\prox_{\iota_{\bd}})\left(\frac{1}{\eta}\byk+\bD(2\tbvkk-\bvk)\right)$

5. $\theta_k=\frac{\varrho k}{k+\delta}$

6. $\bvkk=(1+\theta_k)\tbvkk-\theta_k\bvk$

7. $\bykk=(1+\theta_k)\tbykk-\theta_k\byk$

8. $k=k+1$\vspace{0.25em}

\hspace{-1.5em}{\bf until} $\frac{\|\bv^{k}-\bv^{k-1}\|_2}{\|\bv^{k-1}\|_2}\leqs tol$ or $k>MaxIter$.

\hspace{-1.5em}9. $\hat{\bw}=\bv^{k}(1:N)$.

\hspace{-1.5em}{\bf Output}: The portfolio $\hat{\bw}$.

\hspace{-1.5em}\rule[0.5em]{16.6cm}{0.1em}

We remark that the KMPA can be directly extended to solve general constrained optimization models of the form
\begin{equation}\label{mod_general}
\hat{\bx}=\argmin_{\bx\in\bbRm}\left\{f(\bx)+g(\bx)\right\},\quad\st\ \bQ\bx\geqs\bq,
\end{equation}
where $f\in\Gamma_0(\bbRm)$ is differentiable with a Lipschitz continuous gradient, $g\in\Gamma_0(\bbRm)$ has a closed form of its proximity operator, $\bQ\in\bbR^{n\times m}$ and $\bq\in\bbR^{n}$. Model \eqref{mod_MPAERL3} is a special case of model \eqref{mod_general} with $n:=6$, $m:=N+1$, $\bQ:=\bD$, $\bq:=\bd$, and $f$, $g$ given by \eqref{def_fandg}.

\section{Experimental Results}\label{sec:experiment}

In this section, we present the performance of the proposed algorithm. We conduct extensive experiments on 6 benchmark data sets from Kenneth R. French's Data Library\footnote{\url{http://mba.tuck.dartmouth.edu/pages/faculty/ken.french/data_library.html}} (a standard and widely-used data library for long-term PO), named FF25, FF25EU, FF32, FF48, FF100 and FF100MEOP. FF25 contains 25 portfolios (they can also be considered as ``assets'' in our experiments) formed on BE/ME (book equity to market equity) and investment from the US market. FF25EU contains 25 portfolios formed on ME and prior return from the European market. FF32 contains 32 portfolios developed by BE/ME and investment from the US market. FF48 contains 48 industry portfolios from the US market. FF100 contains 100 portfolios formed on ME and BE/ME, while FF100MEOP contains 100 portfolios formed on ME and operating profitability, all from the US market. All these data sets are monthly price relative sequences, which is a conventional frequency setting for long-term PO. Their profiles are shown in Table \ref{tab:infodataset}.

\begin{table}[h]
\footnotesize
\centering
\caption{Information of 6 benchmark data sets from real-world financial markets.}
\label{tab:infodataset}
\vspace{0.5em}
\begin{tabular}{|c|@{}c@{}|c|@{}c@{}|c|}
\hline
Data Set & Region & Time & Months & Assets\\
\hline
FF25 &US & $Jul/1971 \sim May/2023$ & 623 & 25\\
FF25EU &EU & $Nov/1990 \sim May/2023$ & 391 & 25\\
FF32 &US & $Jul/1971\sim May/2023$  & 623 &  32\\
FF48 &US   & $Jul/1971 \sim May/2023$ & 623 & 48\\
FF100 &US & $Jul/1971 \sim May/2023$ & 623 & 100\\
FF100MEOP &US & $Jul/1971 \sim May/2023$ & 623 & 100\\
\hline
\end{tabular}
\end{table}

We compare the proposed MPAERL with 9 state-of-the-art PO models (introduced in Section \ref{sec:relatework}): SSMP \cite{sparsepo}, SSPO \cite{SSPO}, SPOLC \cite{SPOLC}, RPRT \cite{RPRT}, S1, S2, S3 \cite{SSPOl0}, MCFPS(I) and MCFPS(II) \cite{gong2022multi}, as well as 2 trivial baseline models: 1/N \cite{1Nstrategy} and Market \cite{olpsjmlr}. S1, S2 and S3 are $3$ slightly different algorithms that solve (\ref{eqn:sparsemodell0}) and (\ref{eqn:sparsemodell0-sol}), in which S1 is deterministic but S2 and S3 are randomized. Additionally,  \cite{gong2022multi} employs a genetic algorithm with inherent randomness to address the MCFPS(I) and MCFPS(II) models.  Thus we run S2, S3 and the genetic algorithm used to solve model MCFPS(I) and MCFPS(II) for 10 times and report their average results in this section. The 1/N strategy rebalances the portfolio to be equally weighted on each trading period, while the Market strategy sets an equally weighted portfolio at the beginning and does not rebalance till the end.

We adopt the moving-window trading framework \cite{egrmvgap} in the experiments, which is consistent with practical portfolio management. In brief, a window size $T$ and the initial wealth $S^{(0)}=1$ are preset for a strategy, then the price relatives in the time window $t=[1:T]$ are used to update the portfolio $\hat{\bw}^{(T+1)}$ for the next trading period. Then we proceed to $(T+1)$ and update the cumulative wealth $S^{(T+1)}=(\mathbf{x}^{(T+1)}\cdot \hat{\bw}^{(T+1)})S^{(T)}$. In the next round, the price relatives in the time window $t=[2:(T+1)]$ are used to update the portfolio $\hat{\bw}^{(T+2)}$, and the above procedure is repeated, till the last period $\mathscr{T}$ of the investment. The equal-weight portfolio can be used at the beginning where there are insufficient samples to run a strategy. By this way, we obtain a backtest sequence $\{S^{(t)}\}_{t=0}^{\mathscr{T}}$ of cumulative wealths, which can be used to compute several evaluation scores for the investing performance and the risk assessment.

\subsection{Parameter Setting}
Before setting the parameters, we conduct sensitivity analyses for the parameters in MPAERL by using two important evaluation indicators Cumulative Wealth (CW) and Sharpe Ratio (SR), whose definitions are provided later in Section \ref{subsec:CW} and \ref{subsec:SR}, respectively. Table \ref{tab:sensreguparam} presents the CW and SR of MPAERL under various regularization parameters, indicating that our model is not sensitive to the regularization parameter around 1, thus we casually set $\tau=1$. Table \ref{tab:CWSRrho1} shows the CW and SR of MPAERL with different lower bounds, which demonstrate that the CW and SR obtained by MPAERL exhibit a certain degree of variation with changes in the lower bound $\rho_1$. After comprehensive consideration of the performance of both CW and SR metrics, we select $\rho_1=0.03$ as the lower bound of expected return level for all subsequent experiments. Table \ref{tab:CWSRrho2} shows that MPAERL is not sensitive to $\rho_2$ around 0.1. So we casually set $\rho_2=0.1$ in the subsequent experiments.

\begin{table*}[h]
\setlength{\tabcolsep}{1.3mm}
\footnotesize
\centering
\caption{Cumulative wealths and Sharpe ratios of MPAERL with different regularization parameters.}
\label{tab:sensreguparam}
\vspace{0.5em}
\begin{tabular}{|c|c|c|c|c|c|c|c|c|c|c|c|c|}
\hline
\multirow{2}{*}{$\tau$} & \multicolumn{2}{c|}{FF25} & \multicolumn{2}{c|}{FF25EU} &
\multicolumn{2}{c|}{FF32} & \multicolumn{2}{c|}{FF48} & \multicolumn{2}{c|}{FF100} & \multicolumn{2}{c|}{FF100MEOP}\\
\cline{2-13}
&CW &SR &CW &SR &CW &SR &CW &SR &CW &SR &CW &SR\\
\hline
0.01&1286.72 &0.2513 &94.64 &0.2505 &1331.23 &0.2483 &1971.45 &0.2644 &1231.69 &0.2454 &1015.61 &0.2431\\
0.1&1013.23 &0.2425 &105.31 &0.2545 &1752.29 &0.2535 &2312.21 &0.2491 &1758.28 &0.2508 &1578.15 &0.2504\\
0.3&994.28 &0.2421 &102.60 &0.2531 &1814.57 &0.2555 &2342.43 &0.2493 &1778.29 &0.2503 &1578.73 &0.2504\\
0.5&1001.95 &0.2424 &102.51 &0.2530 &1805.33 &0.2554 &2343.12 &0.2493 &1777.91 &0.2503 &1577.52 &0.2504\\
0.7&995.77 &0.2422 &102.47 &0.2530 &1802.72 &0.2555 &2343.22 &0.2493 &1777.44 &0.2503 &1555.63 &0.2502\\
0.9&999.77 &0.2423 &99.18 &0.2511 &1793.55 &0.2552 &2343.41 &0.2493 &1776.90 &0.2503 &1555.63 &0.2502\\
1 &998.54 &0.2423 &102.66 &0.2531 &1802.79 &0.2553 &2343.57 &0.2493 &1776.51 &0.2503 &1578.05 &0.2504\\
1.2&1000.16 &0.2423 &102.46 &0.2530 &1802.67 &0.2553 &2344.02 &0.2493 &1775.60 &0.2503 &1555.67 &0.2504\\
1.5&998.63 &0.2423 &99.24 &0.2512 &1793.50 &0.2552 &2344.82 &0.2493 &1774.58 &0.2502 &1577.72 &0.2504\\
2&997.80 &0.2423 &102.64 &0.2531 &1814.88 &0.2555 &2345.90 &0.2493 &1772.80 &0.2502 &1555.52 &0.2504\\
3&996.43 &0.2422 &102.79 &0.2532 &1794.18 &0.2552 &2348.22 &0.2493 &1770.18 &0.2502 &1555.32 &0.2502\\
4&996.67 &0.2422 &102.91 &0.2532 &1795.00 &0.2552 &2348.49 &0.2493 &1768.99 &0.2501 &1576.10 &0.2502\\
\hline
\end{tabular}
\end{table*}

Since SSMP and MPAERL are both based on the Markowitz's criterion, we empirically set the same regularization parameter $\tau$ and tune the same window size $T=18$ for these two methods. The expected return $\rho$ in SSMP is set as $0.066$ according to \cite{sparsepo}.
Based on the convergence analysis of KMPA in Section \ref{mainsec}, we always let $\delta=3$, $\varrho=0.8$, $\xi=1-\varrho$, $L=\frac{2}{T}\|\tilde{\bR}^\top\tilde{\bR}\|_2$,
$\beta=\frac{\xi}{L}$ and $\eta=\frac{\xi(2\xi-\beta L)}{4\beta\xi^2\|\bD\|_2^2+L(2\xi-\beta L)}$. We repeat Algorithm 1 until the equality tolerance $\frac{\|\bv^{k}-\bv^{k-1}\|_2}{\|\bv^{k-1}\|_2}<10^{-8}$
or the maximum iteration number $10,000$ is reached. For the MCFPS(I) and MCFPS(II) models, the invested proportion of the risk-free asset was set to 0, while the remaining parameters were set as in \cite{gong2022multi}.
Additionally, based on their performance on the CW, an evaluation indicator to be introduced in the next subsection, the parameter $a$ was set to 1 for the MCFPS(I) model and 1.5 for the MCFPS(II) model.
As for other compared methods, we set their parameters by the defaults in their original papers.

\begin{table*}[htbp]
\setlength{\tabcolsep}{1.0mm}
\footnotesize
\centering
\caption{Cumulative wealths and Sharpe ratios of MPAERL with different lower bounds of the expected return level.}
\label{tab:CWSRrho1}
\vspace{0.5em}
\begin{tabular}{|c|c|c|c|c|c|c|c|c|c|c|c|c|}
\hline
\multirow{2}{*}{$Bound$} & \multicolumn{2}{c|}{FF25} & \multicolumn{2}{c|}{FF25EU} &
\multicolumn{2}{c|}{FF32} & \multicolumn{2}{c|}{FF48} & \multicolumn{2}{c|}{FF100} & \multicolumn{2}{c|}{FF100MEOP}\\
\cline{2-13}
&CW &SR &CW &SR &CW &SR &CW &SR &CW &SR &CW &SR\\
\hline
0.01$\sim$0.1 &272.23 &0.2305 &13.89 &0.1692 &358.54 &0.2472  &167.23 &0.2290 &468.94 &0.2402 &300.88 &0.2295\\
0.02$\sim$0.1 &610.75 &0.2427 &41.47 &0.2179 &989.20 &0.2611 &682.84 &0.2518 &1439.66 &0.2643 &942.43 &0.2523\\
0.03$\sim$0.1 &998.54 &0.2423 &102.66 &0.2531 &1802.79 &0.2553 &2343.57 &0.2493 &1776.51 &0.2503 &1578.05 &0.2504\\
0.04$\sim$0.1 &1537.01 &0.2474 &96.41 &0.2420 &2393.61 &0.2485 &3131.85 &0.2255 &1445.58 &0.2258 &1826.97 &0.2382\\
0.05$\sim$0.1 &1347.24 &0.2395 &105.71 &0.2424 &1231.32 &0.2254 &2510.55 &0.2028 &827.92 &0.2034  &1033.22 &0.2140\\
0.06$\sim$0.1 &1416.94 &0.2405 &104.34 &0.2407 &1217.17 &0.2233 &1808.75 &0.1869 &589.78 &0.1907 &632.93 &0.1978\\
\hline
\end{tabular}
\end{table*}

\begin{table*}[htbp]
\setlength{\tabcolsep}{1.0mm}
\footnotesize
\centering
\caption{Cumulative wealths and Sharpe ratios of MPAERL with different upper bounds of the expected return level.}
\label{tab:CWSRrho2}
\vspace{0.5em}
\begin{tabular}{|c|c|c|c|c|c|c|c|c|c|c|c|c|}
\hline
\multirow{2}{*}{$Bound$} & \multicolumn{2}{c|}{FF25} & \multicolumn{2}{c|}{FF25EU} &
\multicolumn{2}{c|}{FF32} & \multicolumn{2}{c|}{FF48} & \multicolumn{2}{c|}{FF100} & \multicolumn{2}{c|}{FF100MEOP}\\
\cline{2-13}
&CW &SR &CW &SR &CW &SR &CW &SR &CW &SR &CW &SR\\
\hline
0.03$\sim$0.08 &997.80 &0.2423  &102.60 &0.2531 &1814.49 &0.2555 &2341.75 &0.2493 &1778.88 &0.2503 &1555.14 &0.2502\\
0.03$\sim$0.09 &1009.96 &0.2427 &98.95 &0.2510 &1802.88 &0.2553 &2342.73 &0.2493 &1777.74 &0.2503 &1577.45 &0.2504\\
0.03$\sim$0.1 &998.54 &0.2423 &102.66 &0.2531 &1802.79 &0.2553 &2343.57 &0.2493 &1776.51 &0.2503 &1578.05 &0.2504\\
0.03$\sim$0.11 &1001.48&0.2424 &102.80 &0.2532 &1814.48 &0.2555 &2344.48 &0.2493 &1775.26 &0.2502 &1578.14 &0.2504\\
0.03$\sim$0.12 &993.87 &0.2421 &102.70 &0.2531 &1814.53 &0.2555 &2345.12 &0.2493 &1774.11 &0.2502 &1556.15 &0.2502\\
0.03$\sim$0.13 &999.16 &0.2423 &102.71 &0.2531 &1802.76 &0.2553 &2346.42 &0.2493 &1772.81 &0.2502 &1556.33 &0.2502\\
\hline
\end{tabular}
\end{table*}

\subsection{Cumulative Wealth}\label{subsec:CW}
The cumulative wealth (CW) sequence $\{S^{(t)}\}_{t=0}^{\mathscr{T}}$ is the most important evaluation score for a strategy throughout an investment. We plot the CW sequences for different strategies on the benchmark data sets in Figure \ref{fig:alldatasetCW}. It shows that the proposed MPAERL outperforms other competitors to a large extent in most time of the investment. The final CWs for different strategies are given in Table \ref{tab:finalCW}, which show that MPAERL achieves the highest scores on all the benchmark data sets. Its final CWs are more than doubling the second highest CWs on FF25EU, FF32, FF100 and FF100MEOP. In particular, MPAERL outperforms the 2 trivial strategies 1/N and Market on FF100, where SSMP could not beat them. It indicates that the proposed adaptive expected return level scheme is effective.

\begin{figure*}[htbp]
\centering
\subfigure[FF25]{
\includegraphics[width=0.48\linewidth]{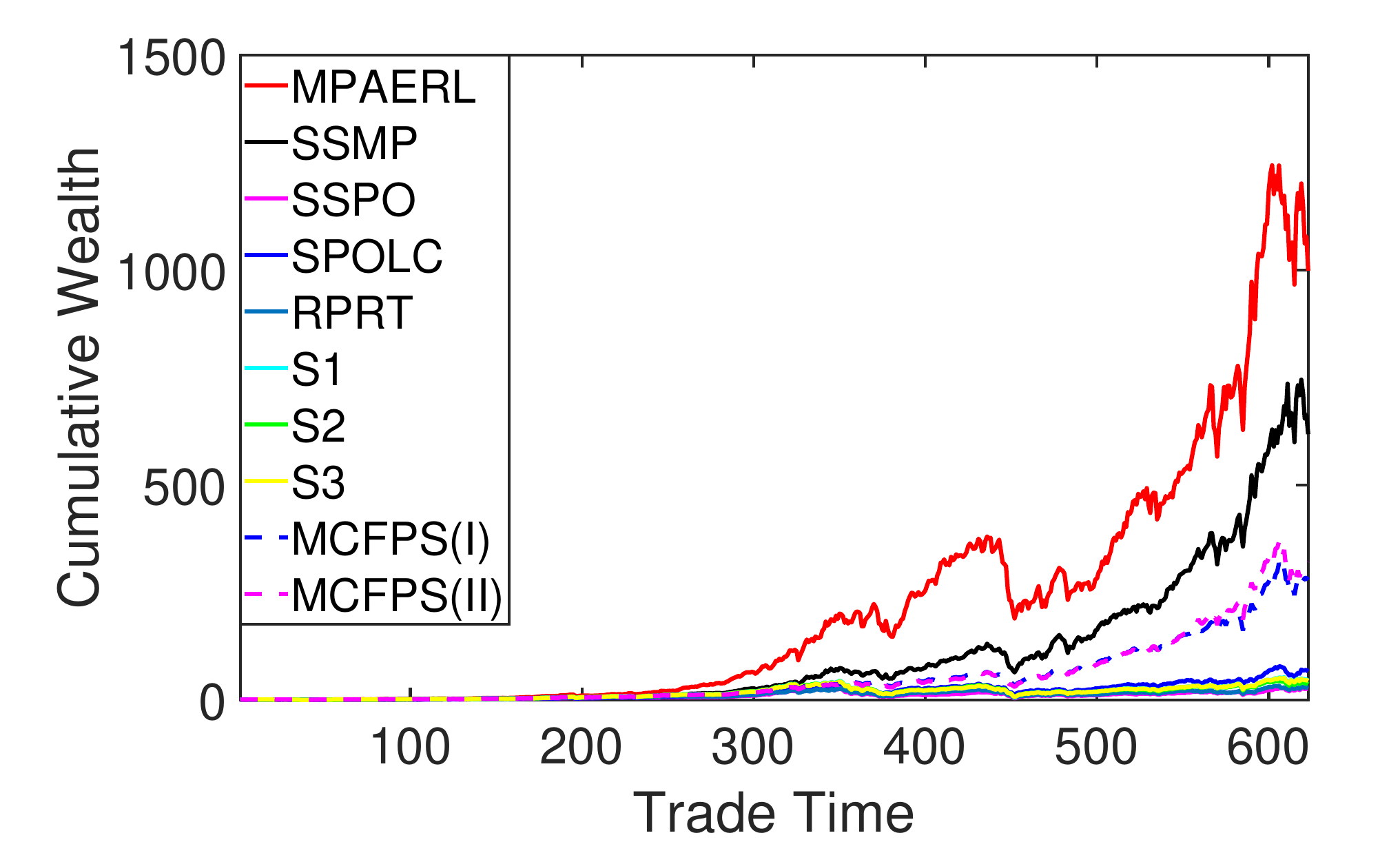}}
\hspace{-1.2em}\subfigure[FF25EU]{
\includegraphics[width=0.48\linewidth]{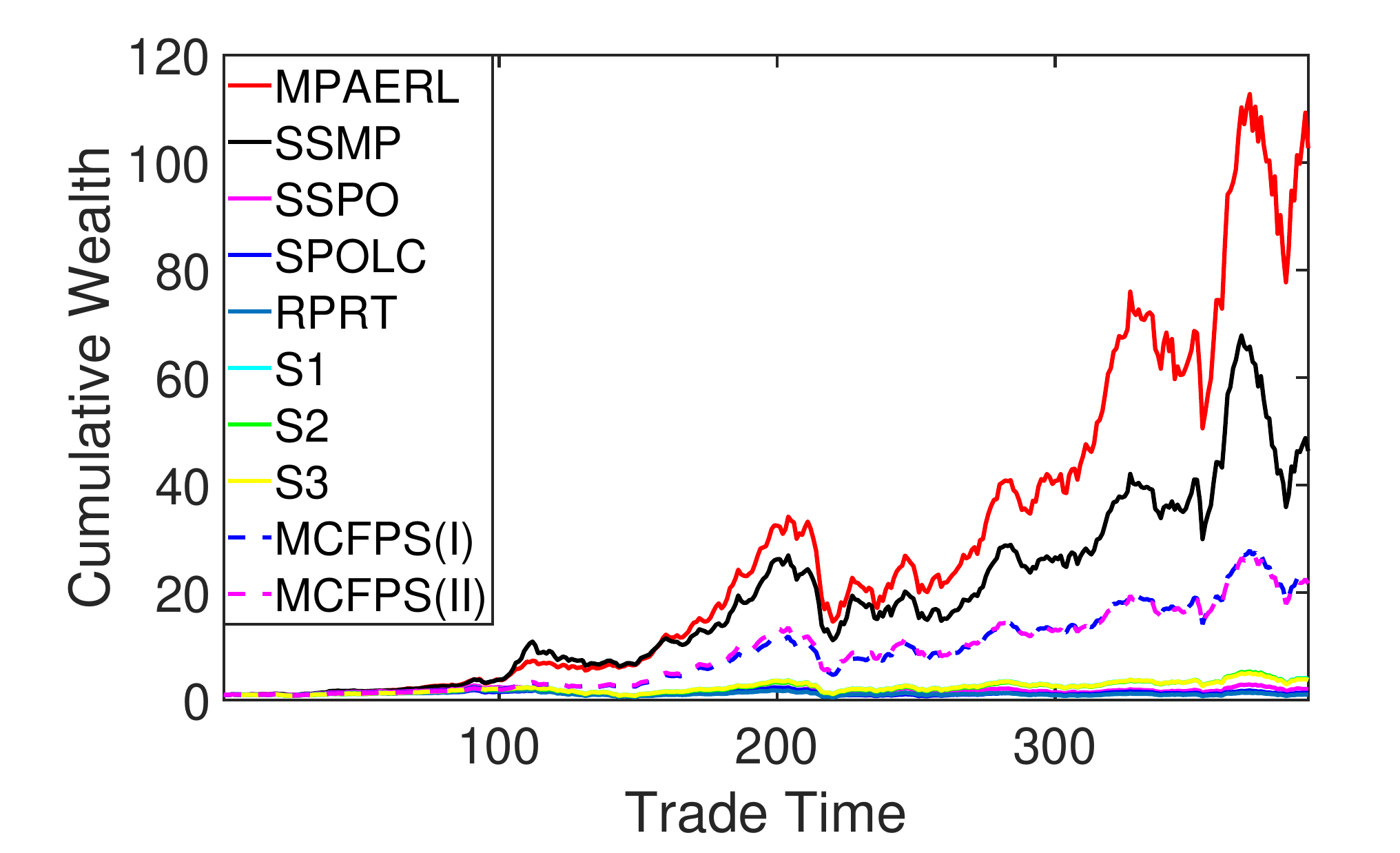}}\\
\subfigure[FF32]{
\includegraphics[width=0.48\linewidth]{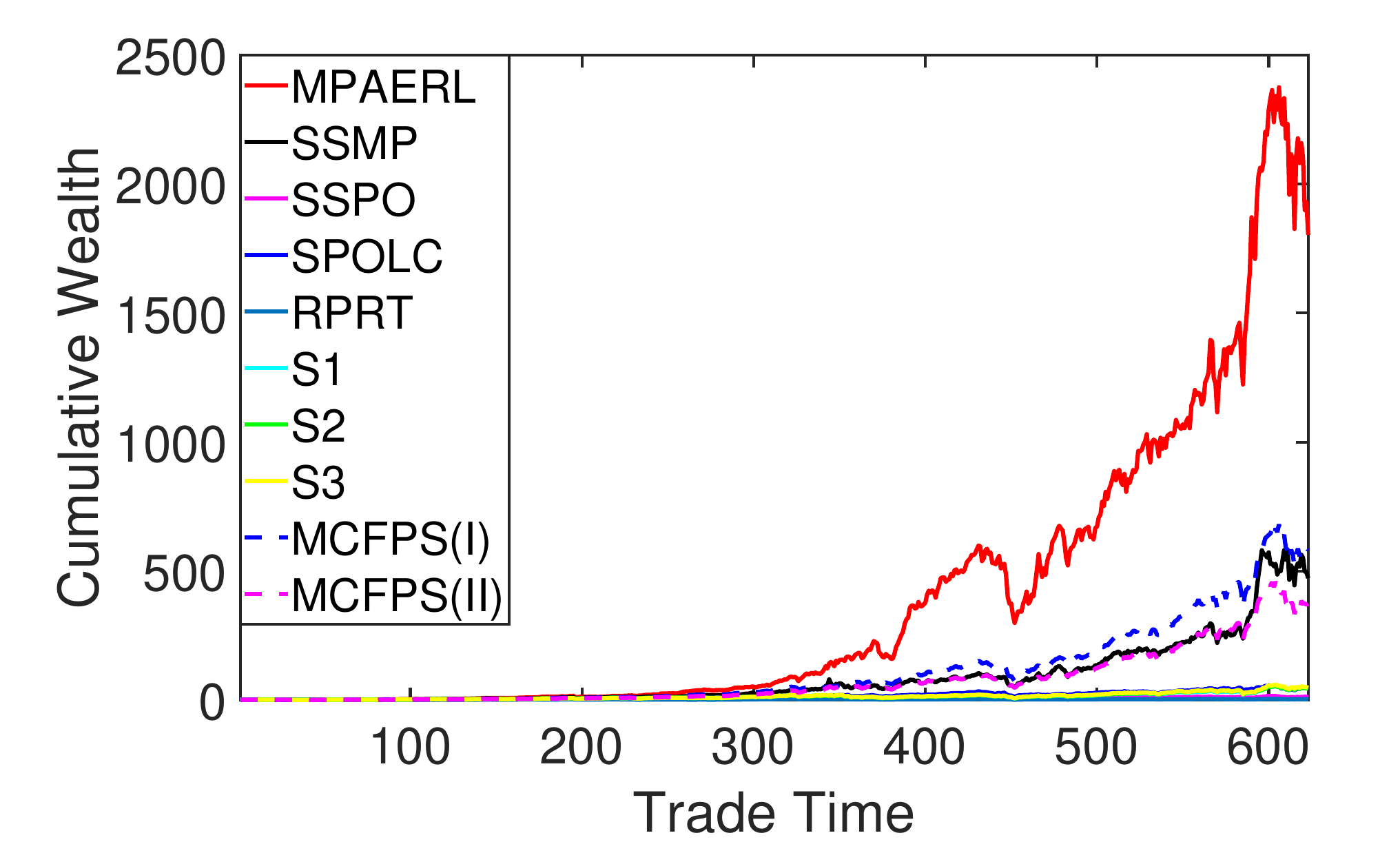}}
\subfigure[FF48]{
\hspace{-1.2em}\includegraphics[width=0.48\linewidth]{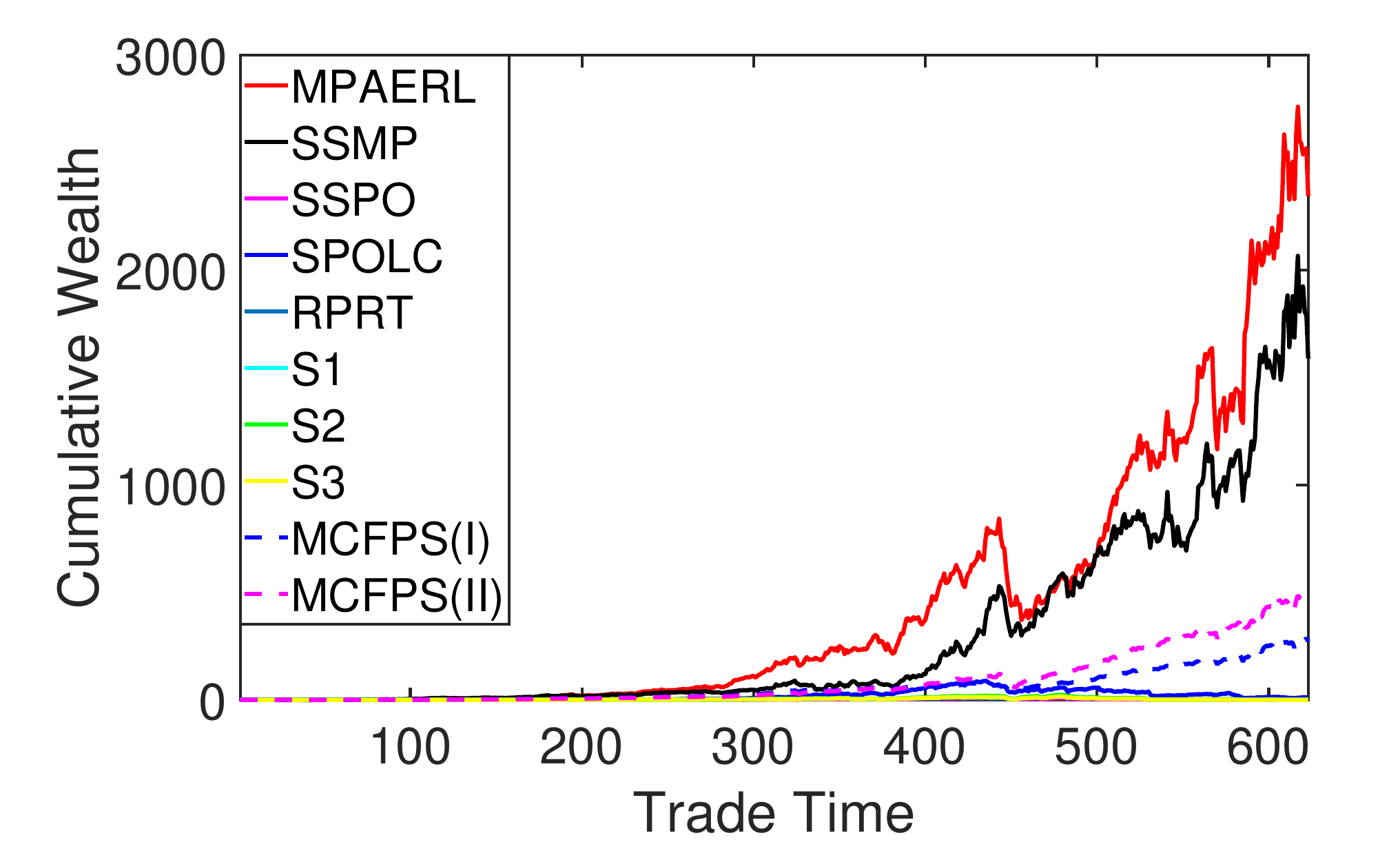}}\\
\subfigure[FF100]{
\includegraphics[width=0.48\linewidth]{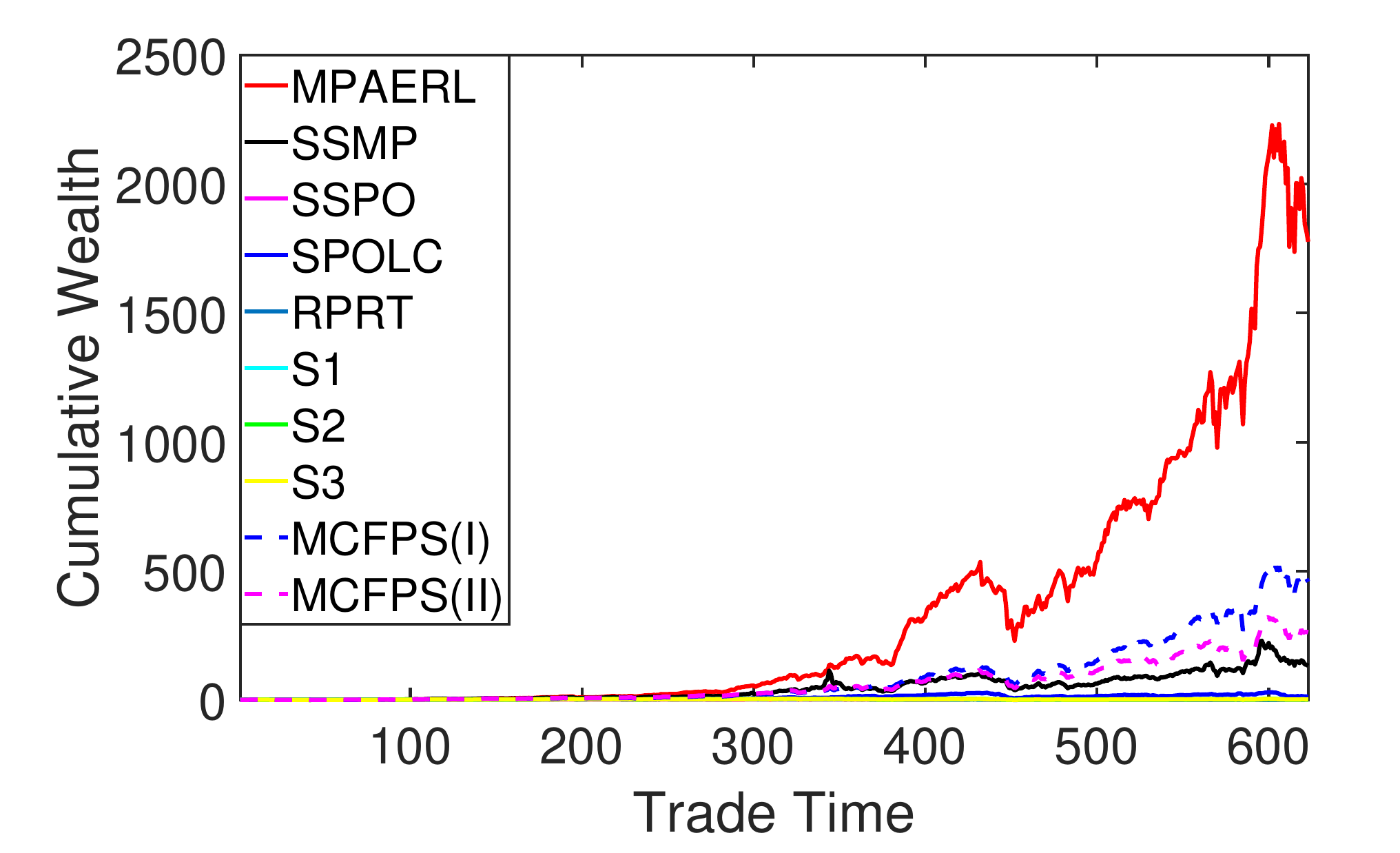}}
\subfigure[FF100MEOP]{
\hspace{-1.2em}\includegraphics[width=0.48\linewidth]{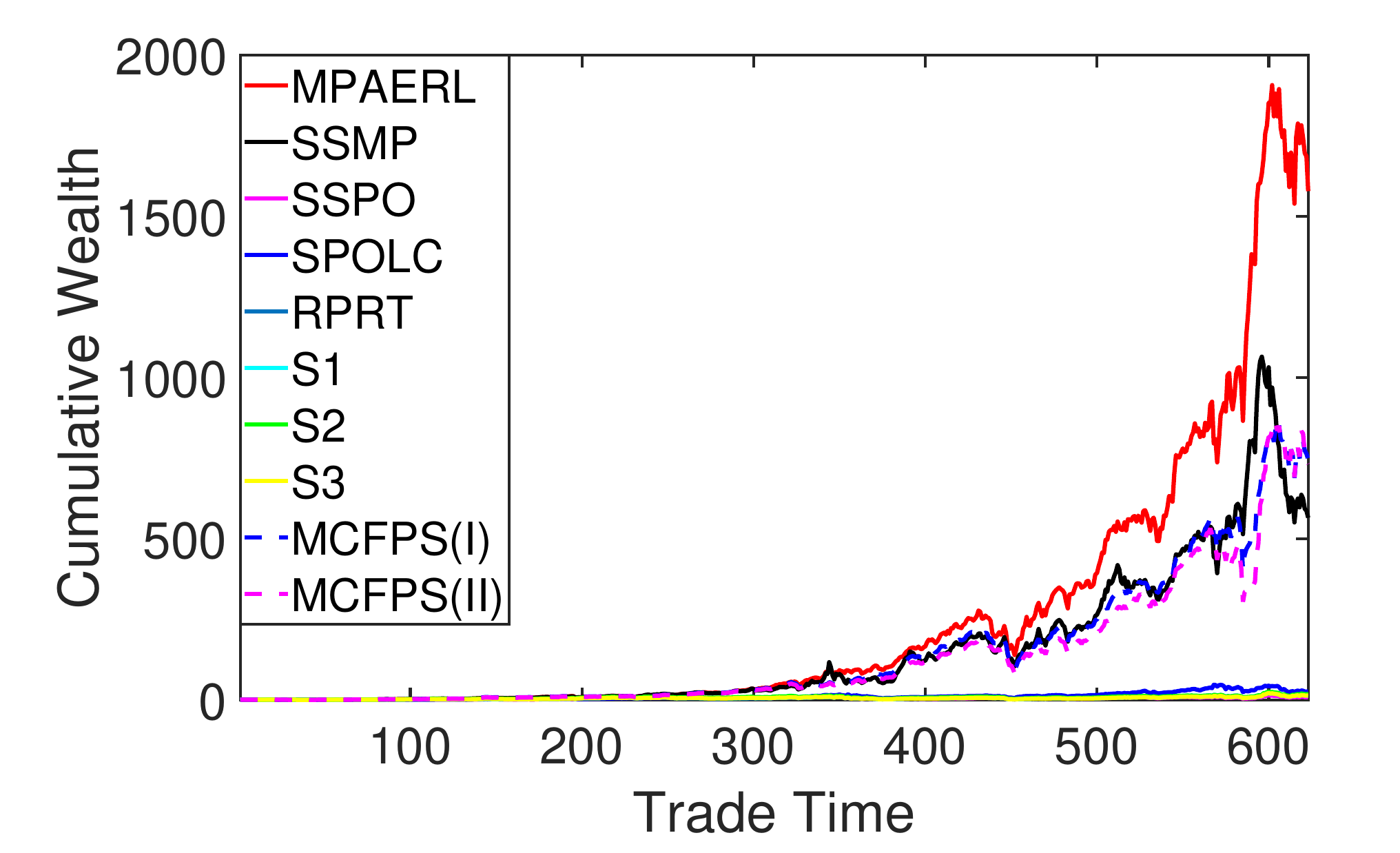}}
\caption{Cumulative wealths of different strategies with respect to trade time on 6 benchmark data sets.}
\label{fig:alldatasetCW}
\end{figure*}

\begin{table}[htbp]
\vspace{1em}
\setlength{\tabcolsep}{1mm}
\footnotesize
\centering
\caption{Final cumulative wealths of different strategies on 6 benchmark data sets.}
\label{tab:finalCW}
\vspace{0.5em}
\begin{tabular}{|c|c|c|c|c|c|c|}
\hline
Strategy & FF25 & FF25EU & FF32 & FF48 & FF100 & FF100MEOP\\
\hline
1/N &355.98 &13.05 &424.42 &235.48 &364.87 &348.70\\
Market &413.79 &43.13 &543.80 &199.85 &505.68 &419.44\\
\hline
SSMP &618.17 &46.34 &472.58 &1588.12 &132.84 &565.56\\
SSPO &25.70 &1.95 &11.58 &0.84 &1.13 &8.70\\
SPOLC &66.71 &1.25 &47.61 &11.90 &13.00 &26.54\\
RPRT &28.30 &0.96 &3.03 &2.34 &0.34 &17.19\\
S1 &45.45 &3.88 &44.13 &0.96 &2.22 &13.96\\
S2 &44.47 &3.91 &40.90 &0.93 &1.94 &7.43\\
S3 &44.81 &3.74 &40.55 &0.96 &2.15 &14.94\\
MCFPS(I) &279.33 &21.96 &581.68 &272.61 &458.42 &746.63\\
MCFPS(II) &321.21 &21.67 &372.83 &471.44 &262.31 &731.91\\
\hline
MPAERL &{\bf998.54} &{\bf102.66} &{\bf1802.79} &{\bf2343.57} &{\bf1776.51} &{\bf1578.05}\\
\hline
\end{tabular}
\end{table}

\newpage
\subsection{$\alpha$ Factor}
In the finance industry, it is also important to evaluate the relative performance of a nontrivial strategy with respect to the Market strategy. The reason is that a portfolio is established from the underlying financial market. If all the asset prices in the financial market drop, the CW cannot rise no matter how we manage the portfolio. In this case, if a nontrivial strategy performs not as badly as the market, it can be considered as effective. Based on the Capital Asset Pricing Model (CAPM) \cite{CAPM}, the $\alpha$ factor \cite{portalpha0} can be used to evaluate this relative performance. Denote ${r}_{s}$ and ${r}_{m}$ as the returns for a nontrivial strategy and the Market strategy, respectively. The $\alpha$ factor can be computed as follows:
\begin{align}
\label{eqn:alphafactor}
&E({r}_{s})=\beta E({r}_{m})+\alpha,\\
\label{eqn:betaestimate}
\notag&\hspace{-2em}\hat{\beta}=\frac{\hat{c}({r}_{s},{r}_{m})}{\hat{\sigma}^2({r}_{m})},\quad \hat{\alpha}=\bar{r}_{s}-\hat{\beta}\bar{r}_{m},
\end{align}
where $E(\cdot)$ denotes the mathematical expectation, $\hat{c}(\cdot,\cdot)$ and $\hat{\sigma}(\cdot)$ denote the sample covariance and the sample standard deviation (STD) computed on the $\mathscr{T}$ trading months, respectively. $\bar{r}_{s}$ denotes the sample mean, which can be computed by $\bar{r}_{s}=\frac{1}{\mathscr{T}}\sum_{t=1}^\mathscr{T} {r}_{s}^{(t)}$ where ${r}_{s}^{(t)}=S_{s}^{(t)}/S_{s}^{(t-1)}-1$ and $S_{s}^{(t)}$ is the CW of this nontrivial strategy on the $t$-th trading month. $\bar{r}_{m}$ can be computed likewise. Since (\ref{eqn:alphafactor}) is essentially a linear regression model, a right-tailed t-test can be implemented to see whether $\alpha$ is significantly greater than $0$. If so, this nontrivial strategy is significantly better than the Market strategy.

The $\alpha$ factors and the p-values for different strategies are given in Table \ref{tab:alphafactor}. MPAERL outperforms all the competitors to a large extent on all the benchmark data sets. Besides, it is the only nontrivial strategy that achieves positive $\alpha$ factors on all the data sets. Moreover, its p-values are all smaller than $0.02$, which indicates that its $\alpha$ factors are greater than $0$ at a confidence level of $98\%$ on all $6$ data sets. To summarize, MPAERL outperforms other state-of-the-art competitors and the Market strategy significantly on the $\alpha$ factor.

\renewcommand\arraystretch{1.1}
\begin{table*}[htbp]
\setlength{\tabcolsep}{0.5mm}
\footnotesize
\centering
\caption{$\alpha$ Factors (with p-values of t-tests) of different strategies on 6 benchmark data sets.}
\label{tab:alphafactor}
\vspace{0.5em}
\begin{tabular}{|c|c|c|c|c|c|c|c|c|c|c|c|c|}
\hline
\multirow{2}{*}{Strategy} & \multicolumn{2}{c|}{FF25} & \multicolumn{2}{c|}{FF25EU} &
\multicolumn{2}{c|}{FF32} & \multicolumn{2}{c|}{FF48} & \multicolumn{2}{c|}{FF100} & \multicolumn{2}{c|}{FF100MEOP}\\
\cline{2-13}
&$\alpha$ &p-value &$\alpha$ &p-value &$\alpha$ &p-value &$\alpha$ &p-value &$\alpha$ &p-value &$\alpha$ &p-value\\
\hline

SSMP &0.0013 &0.1025 &0.0002 &0.4392 &-0.0004 &0.6155 &0.0053 &0.0087 &-0.0019 &0.8778 &0.0006 &0.3542\\
SSPO &-0.0053 &0.9997 &-0.0089 &1.0000 &-0.0065 &1.0000 &-0.0073 &0.9786 &-0.0114 &1.0000 &-0.0063 &0.9967\\
SPOLC &-0.0030 &0.9879 &-0.0093 &1.0000 &-0.0030 &0.9850 &-0.0034 &0.8982 &-0.0054 &0.9969 &-0.0036 &0.9735\\
RPRT &-0.0051 &0.9991 &-0.0106 &1.0000 &-0.0086 &1.0000 &-0.0065 &0.9608 &-0.0128 &1.0000 &-0.0048 &0.9782\\
S1 &-0.0047 &0.9999 &-0.0069 &1.0000 &-0.0050 &1.0000 &-0.0080 &0.9921 &-0.0104 &1.0000 &-0.0060 &0.9988\\
S2 &-0.0049 &0.9975 &-0.0069 &0.9998 &-0.0051 &0.9982 &-0.0080 &0.9822 &-0.0108 &1.0000 &-0.0068 &0.9982\\
S3 &-0.0048 &0.9997 &-0.0070 &1.0000 &-0.0051 &0.9999 &-0.0080 &0.9917 &-0.0104 &1.0000 &-0.0059 &0.9979\\
MCFPS(I) &0.0006 &0.2189 &-0.0017 &0.9523 &0.0013 &0.0500 &0.0021 &0.0176 &0.0005 &0.3258 &0.0016 &0.0377\\
MCFPS(II) &0.0006 &0.2493 &-0.0018 &0.9734 &0.0001 &0.4624 &0.0025 &0.0231 &-0.0009 &0.7497 &0.0013 &0.1323\\
\hline
MPAERL &{\bf0.0024} &0.0151 &{\bf0.0023} &0.0117 &{\bf0.0033} &0.0015 &{\bf0.0059} &0.0002 &{\bf0.0036} &0.0015 &{\bf0.0035} &0.0007\\
\hline
\end{tabular}
\end{table*}

\subsection{Sharpe Ratio}\label{subsec:SR}
Besides return, an investor should also consider the risk of the portfolio. The sample STD of the portfolio return $\hat{\sigma}({r}_{s})$ is a basic risk measurement in the finance industry. Furthermore, Sharpe Ratio (SR) \cite{SHARPratio} is a kind of risk-adjusted return based on CAPM:
\begin{equation*}
\label{eqn:sharperatio}
SR=\frac{\bar{r}_{s}-{r}_{f}}{\hat{\sigma}({r}_{s})},
\end{equation*}
where ${r}_{f}$ denotes the return of some risk-free asset. Since we do not consider risk-free assets in this paper, we let ${r}_{f}=0$. Then SR becomes a quotient of return over risk.

The (monthly) SRs of different strategies are shown in Table \ref{tab:sharpratio}. Note that we need not necessarily annualize the SRs to make comparisons, thus we directly present the computed monthly SRs. The results show that the 2 trivial strategies 1/N and Market outperform other state-of-the-art competitors except the MCFPS and the proposed MPAERL on 4 data sets and 6 data sets, respectively. The reason is that these 2 trivial strategies aim to diversify the risk over all the assets, which is essentially a risk control scheme. Previous researches \cite{1Nstrategy} also verify that such trivial strategies are very competitive in the risk-adjusted return. Moreover, while MCFPS(I) surpasses the 2 trivial strategies 1/N and Market on 3 data sets, our MPAERL outperforms both trivial and nontrivial strategies on all data sets.  MPAERL not only allows for an adaptive expected return level but also reduces the risk at this level, and this return-risk balance can be dynamically adaptive to the ever-changing financial market.

\begin{table}[htbp]
\setlength{\tabcolsep}{1mm}
\footnotesize
\centering
\caption{Sharpe Ratios of different strategies on 6 benchmark data sets.}
\label{tab:sharpratio}
\vspace{0.5em}
\begin{tabular}{|c|c|c|c|c|c|c|c|}
\hline
Strategy & FF25 & FF25EU & FF32 & FF48 & FF100 & FF100MEOP\\
\hline
1/N &0.2278 &0.1576 &0.2236 &0.2059 &0.2089 &0.2077\\
Market &0.2287 &0.2254 &0.2240 &0.2093 &0.2187 &0.2132\\
\hline
SSMP &0.2251 &0.2102 &0.1890 &0.2080 &0.1495 &0.1840\\
SSPO &0.1132 &0.0602 &0.0925 &0.0517 &0.0486 &0.0851\\
SPOLC &0.1438 &0.0412 &0.1363 &0.0905 &0.0931 &0.1120\\
RPRT &0.1142 &0.0334 &0.0620 &0.0677 &0.0276 &0.0978\\
S1 &0.1299 &0.0879 &0.1259 &0.0511 &0.0591 &0.0944\\
S2 &0.1212 &0.0860 &0.1176 &0.0478 &0.0524 &0.0809\\
S3 &0.1277 &0.0857 &0.1222 &0.0512 &0.0584 &0.0951\\
MCFPS(I) &0.2241 &0.1861 &0.2404 &0.2272 &0.2066 &0.2272\\
MCFPS(II)&0.2207 &0.1827 &0.2120 &0.2280 &0.1854 &0.2149\\
\hline
MPAERL &{\bf0.2423} &{\bf0.2531} &{\bf0.2553} &{\bf0.2493} &{\bf0.2503} &{\bf0.2504}\\
\hline
\end{tabular}
\end{table}

\subsection{Maximum Drawdown}
In the finance industry, it is important to examine the extreme loss of a strategy during an investment as part of the risk assessment. A widely-used metric is the maximum drawdown (MDD) \cite{MDD} that measures the maximum percentage loss of CW from a peak to a subsequent valley in the whole investment
\begin{equation*}
MDD:=\max\limits_{l\in[1,\mathscr{T}]}\frac{\max\limits_{t\in[1,l]}{S}^{(t)}-{S}^{(l)}}{\max\limits_{t\in[1,l]}{S}^{(t)}}=1-\min\limits_{l\in[1,\mathscr{T}]}\left(\frac{{S}^{(l)}}{\max\limits_{t\in[1,l]}{S}^{(t)}}\right).
\end{equation*}
It lets the current time $l$ pass from $1$ to $\mathscr{T}$, and searches the past time $t\in[1,l]$ for the peak and the valley CWs to compute the maximum percentage loss. Note that the MDD is a nonnegative value, i.e., the absolute value of the actual percentage loss. As the investing time $\mathscr{T}$ increases, MDD would not decrease. Hence it is difficult to keep a relatively low MDD in a long investment.

MDDs of different strategies are shown in Table \ref{tab:maxDD}. MPAERL outperforms other state-of-the-art competitors on 5 out of 6 data sets, which shows a good capability of downside risk control. In general, the risk inevitably increases as the return increases for any strategy, but MPAERL enjoys high CWs while keeping competitive MDDs at the same time. Hence MPAERL is effective in balancing return and risk due to its adaptive expected return level scheme.

\begin{table}[htbp]
\setlength{\tabcolsep}{1mm}
\footnotesize
\centering
\caption{Maximum drawdowns of different strategies on 6 benchmark data sets.}
\label{tab:maxDD}
\vspace{0.5em}
\begin{tabular}{|c|c|c|c|c|c|c|c|}
\hline
Strategy & FF25 & FF25EU & FF32 & FF48 & FF100 & FF100MEOP\\
\hline
SSMP &0.5096 &0.5865 &0.5252 &{\bf0.4683} &0.7083 &0.5653\\
SSPO &0.8456 &0.7570 &0.6848 &0.9587 &0.8586 &0.8427\\
SPOLC &0.6892 &0.7087 &0.6267 &0.9024 &0.7792 &0.7107\\
RPRT &0.8141 &0.7509 &0.7324 &0.9383 &0.9352 &0.7945\\
S1 &0.8439 &0.6920 &0.6716 &0.9661 &0.8735 &0.8195\\
S2 &0.8560 &0.7030 &0.7146 &0.9758 &0.8869 &0.8558\\
S3 &0.8479 &0.6945 &0.6753 &0.9657 &0.8777 &0.8125\\
MCFPS(I) &0.5077 &0.6078 &0.5159 &0.5538 &0.6495 &0.5797\\
MCFPS(II) &0.5317 &0.6256 &0.5250 &0.5632 &0.6246 &0.5602\\
\hline
MPAERL &{\bf0.5012} &{\bf0.5730} &{\bf0.5006} &0.5586 &{\bf0.5703} &{\bf0.4970}\\
\hline
\end{tabular}
\end{table}

\subsection{Transaction Cost}
The transaction cost is an important practical issue for a strategy to be adopted in the real-world investment. We introduce the proportional transaction cost model \cite{UPtc,OLMAR,RMR2,egrmvgap} to fix the CW at the beginning of the $t$-th trading month as follows:
$$
S_\mathscr{T}^{\nu}= S^{(0)}\prod_{t=1}^{\mathscr{T}}[(\mathbf{x}^{(t)}\hat{\bw}^{(t)})\cdot
(1-\frac{\nu}{2}\sum_{i=1}^{N}|\hat{w}_i^{(t)}-\tilde{w}_{i}^{(t-1)}|)],\ \hspace{1em}\tilde{w}_{i}^{(t-1)}= \frac{\hat{w}_{i}^{(t-1)}\cdot\text{x}_{i}^{(t-1)}}{\mathbf{x}^{(t-1)}\hat{\bw}^{(t-1)}},
$$
where $\mathbf{x}^{(t)}$ is the asset price relative vector of the $t$-th month (i.e., the $t$-th row of the price relative matrix $\bX$), $\tilde{w}_{i}^{(t-1)}$ is the adjusted portfolio of Asset $i$ at the end of the $(t-1)$-th month and $\tilde{\bw}^{(0)}$ is set as the vector ${\bm0}_N$. Given the transaction cost rate $\nu\in[0,1)$, the term $\frac{\nu}{2}\sum_{i=1}^{N}|\hat{w}_i^{(t)}-\tilde{w}_{i}^{(t-1)}|$ computes the proportional transaction cost when the adjusted portfolio $\tilde{\bw}^{(t-1)}$ is updated as the next portfolio $\hat{\bw}^{(t)}$.

We let $\nu$ change in $0\sim 0.5\%$ ($0.5\%$ is a rather high transaction cost rate in the real-world finance industry) and plot the final CWs of different strategies in Figure \ref{fig:trancostCW}. It shows that MPAERL outperforms other state-of-the-art competitors on all 6 data sets in all the cases, which suggests that MPAERL is also effective in controlling transaction costs while managing the portfolio adaptively.

\begin{figure*}[htbp]
\centering
\hspace{-0.8em}\subfigure[FF25]{
\includegraphics[width=0.343\linewidth]{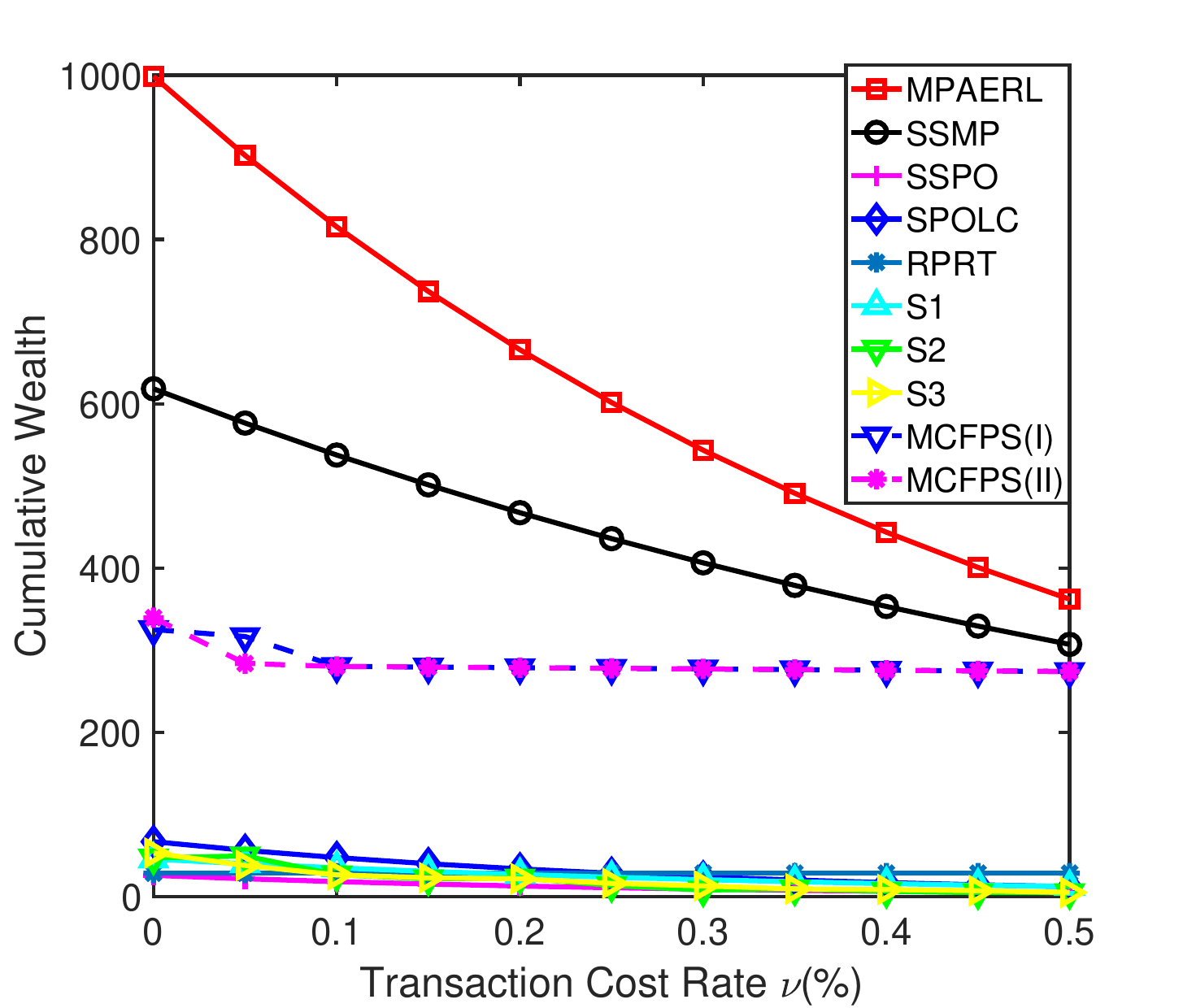}}
\hspace{-1.2em}\subfigure[FF25EU]{
\includegraphics[width=0.343\linewidth]{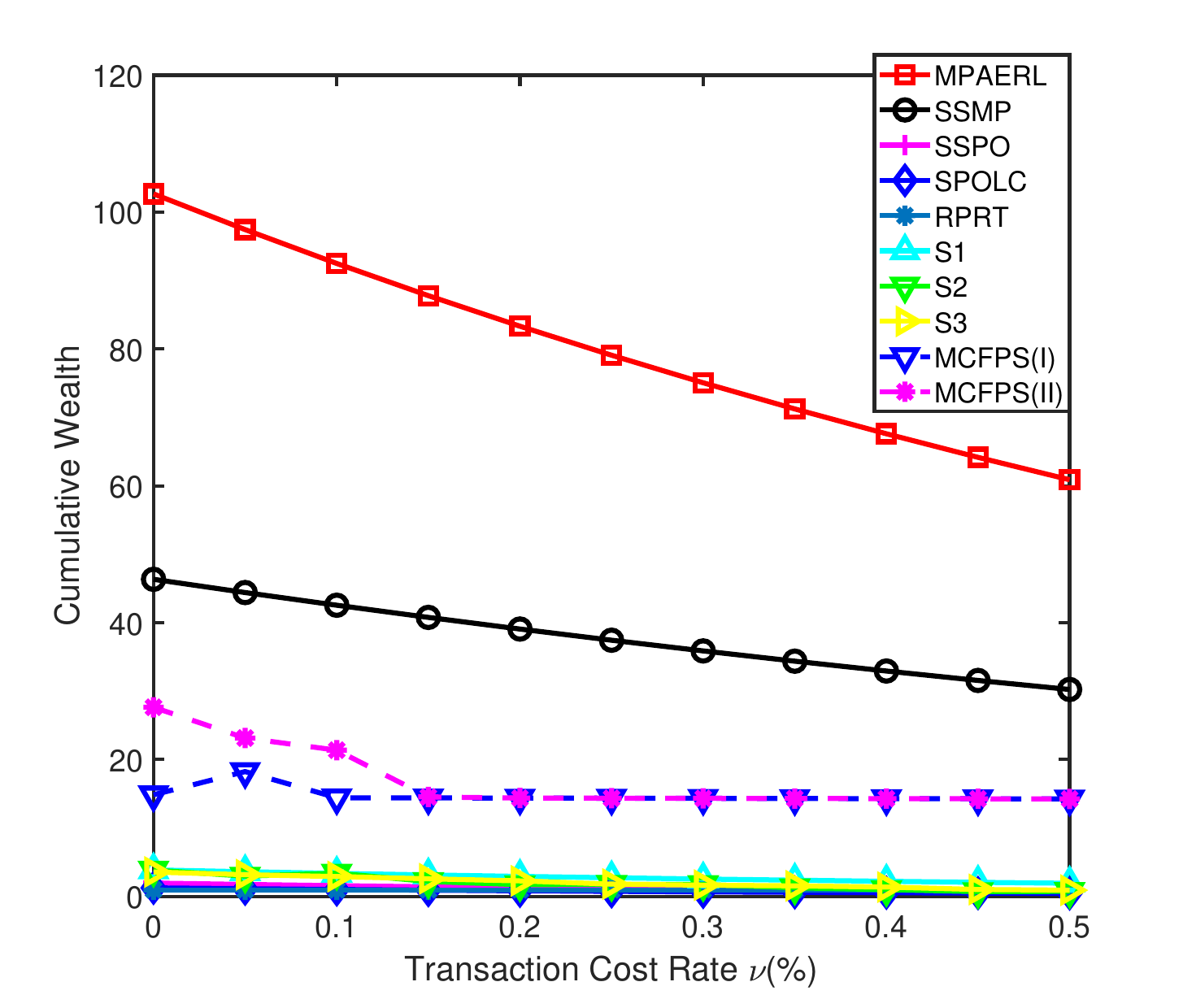}}
\hspace{-1.2em}\subfigure[FF32]{
\includegraphics[width=0.343\linewidth]{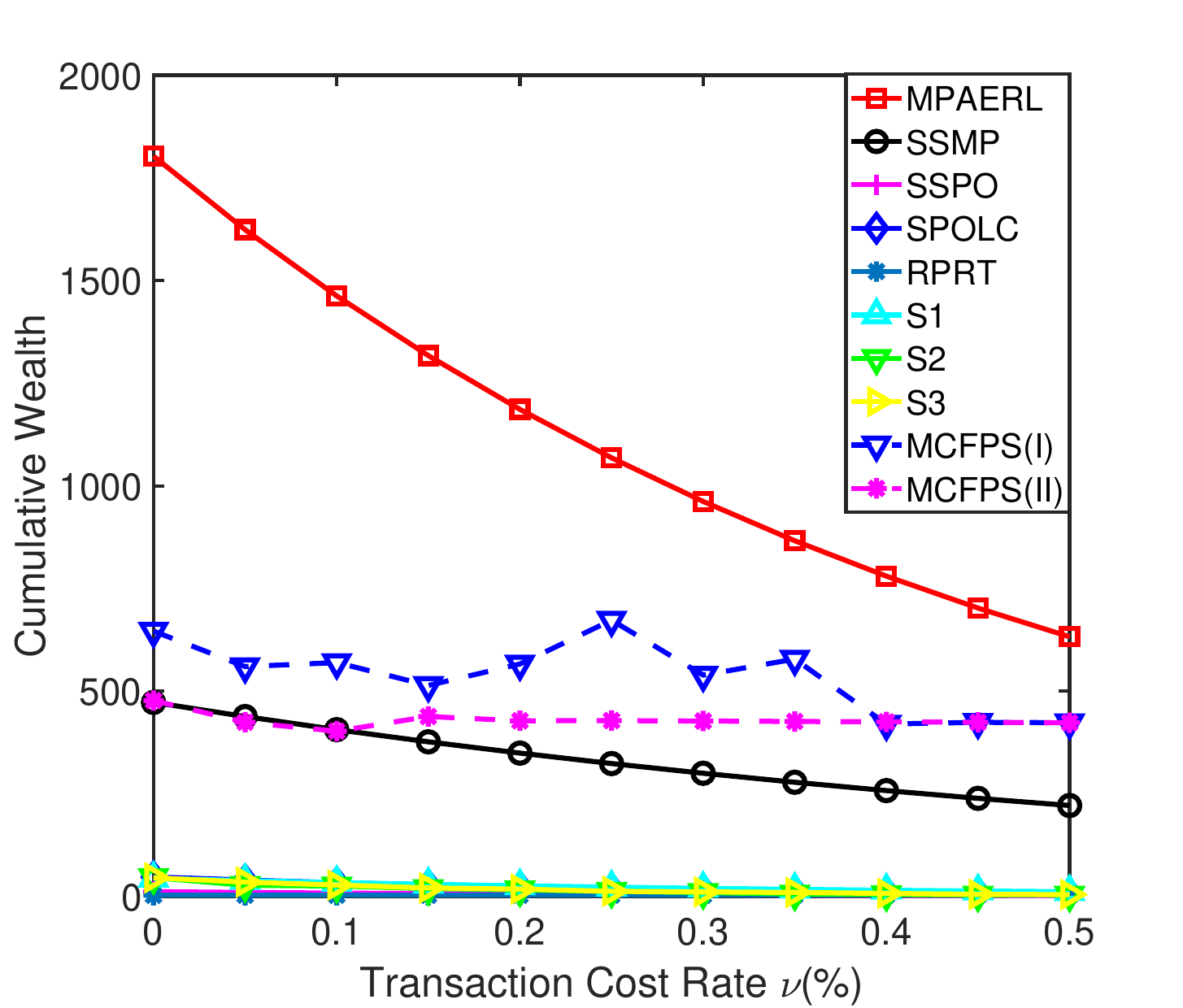}}\\
\hspace{-0.8em}\subfigure[FF48]{
\includegraphics[width=0.343\linewidth]{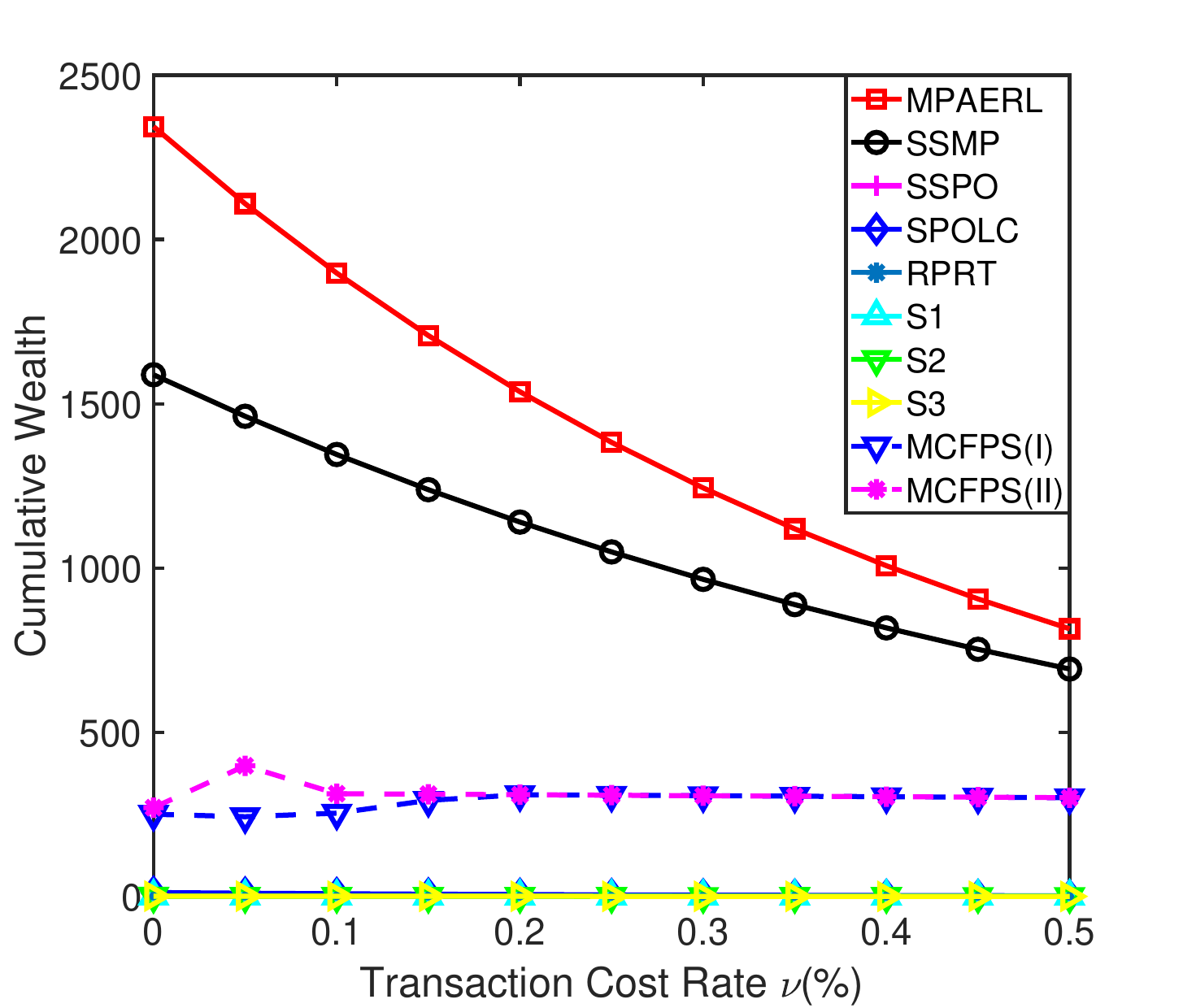}}
\hspace{-1.2em}\subfigure[FF100]{
\includegraphics[width=0.343\linewidth]{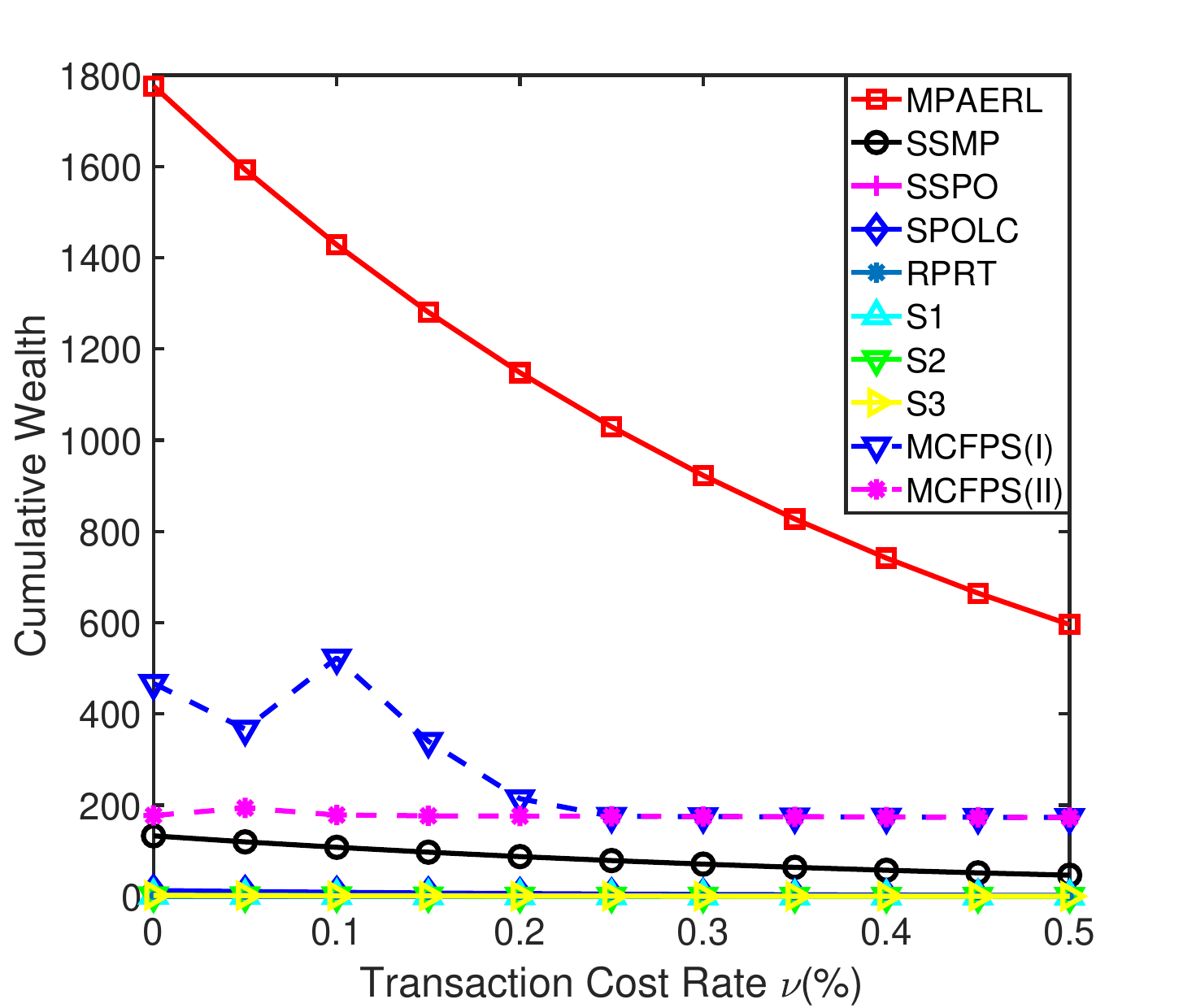}}
\hspace{-1.2em}\subfigure[FF100MEOP]{
\includegraphics[width=0.343\linewidth]{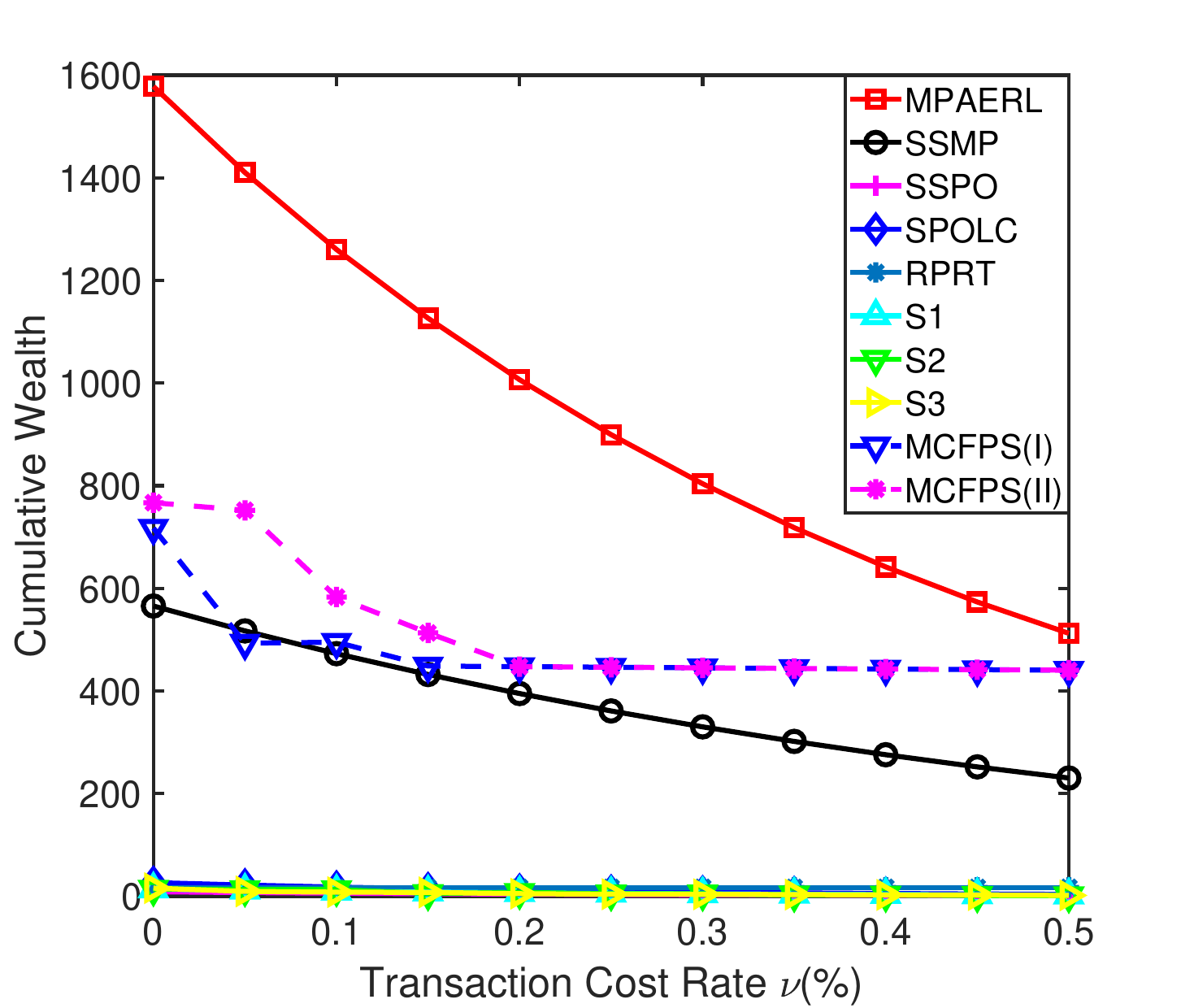}}
\caption{Final cumulative wealths of different strategies with respect to different transaction cost rates on 6 benchmark data sets.}
\label{fig:trancostCW}
\end{figure*}

\section{Conclusion}\label{sec:conclusion}
We propose a novel Markowitz Portfolio with Adaptive Expected Return Level (MPAERL) to improve the traditional return-risk balance scheme in the finance industry. Specifically, traditional portfolio management presets a fixed expected return level according to the preferred risk appetite, then tries to minimize the portfolio risk at this return level. Such a scheme may not favor nonprofessional investors that do not know their risk appetites well, and may not be adaptive to the ever-changing financial market. To fill this gap, we propose to optimize the expected return level and the portfolio simultaneously, in order to dynamically balance the return and the risk of a portfolio. Moreover, we propose an exact, convergent, and efficient Krasnoselskii-Mann Proximity Algorithm (KMPA) based on the proximity operator and the Krasnoselskii-Mann momentum technique to solve the proposed model. KMPA can solve not only the proposed model, but also a general two-term optimization problem with inequality constraints.

Extensive experiments are conducted on $6$ benchmark data sets from the French's widely-used public data library. The results show that MPAERL outperforms other state-of-the-art competitors in several major evaluation scores for investing performance, including the cumulative wealth, the $\alpha$ factor, and the Sharpe Ratio. MPAERL also has a competitive capability of downside risk control according to the maximum drawdown experiments, which indicates that its adaptive scheme can effectively balance return and risk. As for practical issues, MPAERL outperforms other state-of-the-art competitors in most cases of the transaction cost experiments. Therefore, this adaptive expected return level approach merits further exploration, with potential future research efforts focused on developing novel return-risk balancing mechanisms. The limitation of MPAERL may lie in the following aspect. Investing strategies based on mathematical finance assume that assets can be bought or sold according to the market price. But in the real world, the actual transaction price is affected by the impact cost. This may have a little influence in the investing performance.

\section{Appendices}
\subsection{Proof of Proposition \ref{prop_TGwelldef}}\label{sec:appendice1}
In this appendix, we provide the proof of Proposition \ref{prop_TGwelldef}. To this end, we first recall the definition of coercivity and prove in the following lemma that the proximity operator of a function in $\Gamma_0(\bbRm)$ is well-defined. Let $\psi:\bbRm\to[-\infty,+\infty]$. We say that $\psi$ is coercive if $\lim_{\|\bx\|_2\to+\infty}\psi(\bx)=+\infty$.
\begin{lemma}\label{lem:proxwelldef}
If $\psi\in\Gamma_0(\bbRm)$, then for any $\bx\in\bbRm$, $\prox_{\psi}(\bx)$ exists and is unique.
\end{lemma}
\begin{proof}
Let $\widetilde{\psi}(\bu):=\frac{1}{2}\|\bu-\bx\|_2^2+\psi(\bu)$ for $\bu\in\bbRm$. Since $\psi\in\Gamma_0(\bbRm)$ and the quadratic term in $\widetilde{\psi}$ is coercive and strictly convex, we know that $\widetilde{\psi}\in\Gamma_0(\bbRm)$ and it is also coercive and strictly convex. Then the existence and the uniqueness of $\prox_{\psi}(\bx)$ follow from Proposition 11.15 of \cite{bauschke2017convex} and the strict convexity of $\widetilde{\psi}$ immediately.
\end{proof}

We then give the proof of Proposition \ref{prop_TGwelldef} as follows.\vspace{0.5em}

\noindent\textbf{\textit{Proof of Proposition \ref{prop_TGwelldef}.}} We first prove Item $(i)$. Let $\bu_1:=(u_1,u_2,\ldots,u_{m_1})^\top$ and\\ $\bu_2:=(u_{m_1+1},u_{m_1+2},\ldots,u_{m_1+m_2})^\top$ for $\bu\in\bbR^{m_1+m_2}$. For a given vector $\bz:=\bigg(\begin{array}{c}
\bv\\
\by
\end{array}\bigg)$
with $\bv\in\bbR^{m_1}$ and $\by\in\bbR^{m_2}$, the implicit fixed-point equation in \eqref{def_operTG} can be written as
\begin{subequations}
\begin{numcases}{}
\label{TGwelleq1}\bu_1=\prox_{\beta g}(\bv-\beta\bD^\top\by),\\
\label{TGwelleq2}\bu_2=\prox_{\eta\iota_{\bd}^*}(2\eta\bD\bu_1-\eta\bD\bv+\by).
\end{numcases}
\end{subequations}
Since vectors $\bv$ and $\by$ are given, the existence and uniqueness of $\bu_1$ in \eqref{TGwelleq1} follows from the fact $\beta g\in\Gamma_0(\bbR^{m_1})$ and Lemma \ref{lem:proxwelldef}. Now that a unique $\bu_1$ is given, to prove the existence and uniqueness of $\bu_2$ in \eqref{TGwelleq2}, it suffices to show that $\eta\iota_{\bd}^*\in\Gamma_0(\bbR^{m_2})$, which follows from Corollary 13.38 of \cite{bauschke2017convex} and the fact $\iota_{\bd}\in\Gamma_0(\bbR^{m_2})$. In conclusion, for any given $\bz\in\bbR^{m_1+m_2}$ in the equation contained in \eqref{def_operTG}, there exists a unique solution $\bu$.

We next prove Item $(ii)$.
\begin{align*}
&\bz\in\Fix(\mT_{\bW})\Leftrightarrow\bz=\mT_{\bG}\left(\bz-\bW^{-1}\nabla r(\bz)\right)\\
\Leftrightarrow\ &\bz=\mF\left((\bE-\bG)\bz+\bG\left(\bz-\bW^{-1}\nabla r(\bz)\right)\right)\\
\Leftrightarrow\ &\bz\in\Fix(\mT_{\beta,\eta}).
\end{align*}
The third equivalence above holds since the definition of $\bW$ in \eqref{defmatW} implies that $\bG\bW^{-1}=\bP$. This completes the proof.
\hspace*{\fill}~\QED

\subsection{Proof of Proposition \ref{prop1_TWaver}}\label{sec:appendice2}

In this appendix, we provide the proof of Proposition \ref{prop1_TWaver}. To this end, we first recall the Baillon-Haddad theorem \cite{baillon1977quelques} and Proposition 2.4 of \cite{combettes2015compositions} as the following Lemma \ref{BHthm} and Lemma \ref{CYavercomp}, and then prove the firm nonexpansiveness of operator $\mF$ in Lemma \ref{lem_mFfirm}.
\begin{lemma}\label{BHthm}
Suppose that $\psi:\bbRm\to\bbR$ is a differentiable convex function. Then $\nabla\psi$ is $L$-Lipschitz for some $L>0$ if and only if
$$
\|\nabla\psi(\bx)-\nabla\psi(\by)\|_2^2\leqs L\langle\bx-\by,\nabla\psi(\bx)-\nabla\psi(\by)\rangle,
$$
for all $\bx,\by\in\bbRm$.
\end{lemma}

\begin{lemma}\label{CYavercomp}
Let $\bH\in\bbR^{m\times m}$ be a symmetric positive definite matrix, $\alpha_1,\alpha_2\in(0,1)$. If $\mT_1:\bbRm\to\bbRm$ and $\mT_2:\bbRm\to\bbRm$ are $\alpha_1$-averaged nonexpansive and $\alpha_2$-averaged nonexpansive with respect to $\bH$, respectively, then $\mT_1\circ\mT_2$ is $\frac{\alpha_1+\alpha_2-2\alpha_1\alpha_2}{1-\alpha_1\alpha_2}$-averaged nonexpansive with respect to $\bH$.
\end{lemma}

\begin{lemma}\label{lem_mFfirm}
Let $\mF:\bbR^{m_1+m_2}\to\bbR^{m_1+m_2}$ and $\bP$ be defined by \eqref{def_randmF} and \eqref{def_EandP}, respectively. Then $\mF$ is firmly nonexpansive with respect to $\bP^{-1}$.
\end{lemma}
\begin{proof}
Since $\beta g\in\Gamma_{0}(\bbR^{m_1})$ and $\eta\iota_{\bd}^*\in\Gamma_{0}(\bbR^{m_2})$, it follows from Lemma 2.4 of \cite{combettes2005signal} that $\prox_{\beta g}$ and $\prox_{\eta\iota_{\bd}^*}$ are both firmly nonexpansive with respect to $\bI$. For $\bu:=\left(\begin{array}{c}
\bu_1\\
\bu_2
\end{array}\right)$ with $\bu_1\in\bbR^{m_1}$, $\bu_2\in\bbR^{m_2}$ and $\bv:=\left(\begin{array}{c}
\bv_1\\
\bv_2
\end{array}\right)$ with $\bv_1\in\bbR^{m_1}$, $\bv_2\in\bbR^{m_2}$, by letting $\bp_1:=\prox_{\beta g}(\bu_1)-\prox_{\beta g}(\bv_1)$, $\bp_2:=\prox_{\eta\iota_{\bd}^*}(\bu_2)-\prox_{\eta\iota_{\bd}^*}(\bv_2)$,
and using the firm nonexpansiveness of $\prox_{\beta g}$ and $\prox_{\eta\iota_{\bd}^*}$, we have $\|\bp_1\|_2^2\leqs\langle\bp_1,\bu_1-\bv_1\rangle$ and $\|\bp_2\|_2^2\leqs\langle\bp_2,\bu_2-\bv_2\rangle$. Let $\bp:=\bigg(\begin{array}{c}
\bp_1\\
\bp_2
\end{array}\bigg)$. Then
\begin{align*}
&\|\mF(\bu)-\mF(\bv)\|_{\bP^{-1}}^2\\
=&\|\bp\|_{\bP^{-1}}^2=\frac{1}{\beta}\|\bp_1\|_2^2+\frac{1}{\eta}\|\bp_2\|_2^2\\
\leqs&\frac{1}{\beta}\langle\bp_1,\bu_1-\bv_1\rangle+\frac{1}{\eta}\langle\bp_2,\bu_2-\bv_2\rangle\\
=&\langle\bp,\bu-\bv\rangle_{\bP^{-1}}=\langle\mF(\bu)-\mF(\bv),\bu-\bv\rangle_{\bP^{-1}},
\end{align*}
which implies the desired result.
\end{proof}

We are now in a position to prove Proposition \ref{prop1_TWaver}.\vspace{0.5em}

\noindent\textbf{\textit{Proof of Proposition \ref{prop1_TWaver}.}} It is obvious that $\bW$ is symmetric. We know from $\lambda_{\min}(\bW)>\frac{L}{2}>0$ that $\bW$ is positive definite. According to the definition of $\mT_{\bW}$ in \eqref{def_operTW} and Lemma \ref{CYavercomp}, to prove the averaged nonexpansiveness of $\mT_{\bW}$, it suffices to show that $\mT_{\bG}$ and $\mI-\bW^{-1}\nabla r$ are both averaged nonexpansive.

We first show the averaged nonexpansiveness of $\mT_{\bG}$. Let $\bu=\mT_{\bG}(\bx)$, $\bv=\mT_{\bG}(\by)$ for $\bx,\by\in\bbR^{m_1+m_2}$, and $\ba_1=\bG(\bx-\bu)$, $\ba_2=\bG(\by-\bv)$. Then
\begin{equation}\label{eq_uTGz}
\begin{cases}
\bu=\mF\left((\bE-\bG)\bu+\bG\bx\right)=\mF\left(\bE\bu+\ba_1\right),\\ \bv=\mF\left((\bE-\bG)\bv+\bG\by\right)=\mF\left(\bE\bv+\ba_2\right).
\end{cases}
\end{equation}
From Lemma \ref{lem_mFfirm}, we know that $\mF$ is firmly nonexpansive with respect to $\bP^{-1}$, where $\bP$ is defined by \eqref{def_EandP}, which together with \eqref{eq_uTGz} yields that
$$
\|\bu-\bv\|_{\bP^{-1}}^2\leqs\langle\bu-\bv,\bE(\bu-\bv)+(\ba_1-\ba_2)\rangle_{\bP^{-1}},
$$
that is,
\begin{equation}\label{neq1_TGfirm}
\langle\bu-\bv,\ba_1-\ba_2\rangle_{\bP^{-1}}\geqs\langle\bu-\bv,\widetilde{\bE}(\bu-\bv)\rangle,
\end{equation}
where $\widetilde{\bE}:=\bP^{-1}(\bI-\bE)=\left(\begin{array}{cc}
{\bm0} & -\bD^\top\\
\bD & {\bm0}
\end{array}\right)$. Note that $\widetilde{\bE}^\top=-\widetilde{\bE}$. For any $\bz\in\bbR^{m_1+m_2}$,
$$
\langle\bz,\widetilde{\bE}\bz\rangle=\bz^\top\widetilde{\bE}^\top\bz=-\bz^\top\widetilde{\bE}\bz=-\langle\bz,\widetilde{\bE}\bz\rangle,
$$
which implies that $\langle\bz,\widetilde{\bE}\bz\rangle=0$. Then \eqref{neq1_TGfirm} becomes
$$
\langle\bu-\bv,\ba_1-\ba_2\rangle_{\bP^{-1}}\geqs0,
$$
that is, $\|\bu-\bv\|_{\bW}^2\leqs\langle\bu-\bv,\bx-\by\rangle_{\bW}$. Hence $T_{\bG}$ is firmly nonexpansive with respect to $\bW$. It shows in Remark 4.34 of \cite{bauschke2017convex} that firm nonexpansiveness is equivalent to $\frac{1}{2}$-averaged nonexpansiveness. Therefore, $T_{\bG}$ is $\frac{1}{2}$-averaged nonexpansive with respect to $\bW$.

We next show that operator $\mI-\bW^{-1}\nabla r$ is averaged nonexpansive with respect to $\bW$. Let $\alpha:=\frac{L}{2\lambda_{\min}(\bW)}$ and $\widetilde{\mN}:=\mI-\frac{1}{\alpha}\bW^{-1}\nabla r$. Then
$$
\alpha\in(0,1)\ \ \mbox{and}\ \ \mI-\bW^{-1}\nabla r=(1-\alpha)\mI+\alpha\widetilde{\mN}.
$$
By the definition of averaged nonexpansiveness, it suffices to show that $\widetilde{\mN}$ is nonexpansive with respect to $\bW$. For this purpose, we verify that matrix $\frac{2\alpha}{L}\bI-\bW^{-1}$ is positive semi-definite. We note that $\bW^{-1}$ is symmetric positive definite with the maximum eigenvalue $\lambda_{\max}(\bW^{-1})=\frac{1}{\lambda_{\min}(\bW)}$. Then $\frac{2\alpha}{L}\bI-\bW^{-1}$ is symmetric with the minimum eigenvalue
$$
\lambda_{\min}\left(\frac{2\alpha}{L}\bI-\bW^{-1}\right)=\frac{2\alpha}{L}-\frac{1}{\lambda_{\min}(\bW)}=0,
$$
which implies that $\frac{2\alpha}{L}\bI-\bW^{-1}$ is positive semi-definite. It is easy to see from the definition of function $r$ in \eqref{def_randmF} that $r$ is convex and differentiable with an $L$-Lipschitz continuous gradient. By Lemma \ref{BHthm}, for any $\bx,\by\in\bbR^{m_1+m_2}$, defining $\bz:=\nabla r(\bx)-\nabla r(\by)$, we have $\|\bz\|_2^2\leqs L\langle\bx-\by,\bz\rangle$, which together with the positive semi-definiteness of matrix $\frac{2\alpha}{L}\bI-\bW^{-1}$ gives
\begin{align*}
&2\alpha\langle\bx-\by,\bz\rangle-\langle\bz,\bW^{-1}\bz\rangle\\
\geqs&\frac{2\alpha}{L}\|\bz\|_2^2-\langle\bz,\bW^{-1}\bz\rangle\\
=&\left\langle\bz,\left(\frac{2\alpha}{L}\bI-\bW^{-1}\right)\bz\right\rangle\geqs0.
\end{align*}
We then have that
\begin{align*}
&\left\|\widetilde{\mN}\bx-\widetilde{\mN}\by\right\|_{\bW}^2=\left\|(\bx-\by)-\frac{1}{\alpha}\bW^{-1}\bz\right\|_{\bW}^2\\
=&\|\bx-\by\|_{\bW}^2+\frac{1}{\alpha^2}\|\bW^{-1}\bz\|_{\bW}^2-\frac{2}{\alpha}\langle\bx-\by,\bW^{-1}\bz\rangle_{\bW}\\
=&\|\bx-\by\|_{\bW}^2-\frac{1}{\alpha^2}\left(2\alpha\langle\bx-\by,\bz\rangle-\langle\bz,\bW^{-1}\bz\rangle\right)\\
\leqs&\|\bx-\by\|_{\bW}^2,
\end{align*}
that is, $\widetilde{\mN}$ is nonexpansive with respect to $\bW$, and hence $\mI-\bW^{-1}\nabla r$ is $\alpha$-averaged nonexpansive with respect to $\bW$.

Now by employing Lemma \ref{CYavercomp}, we conclude from the $\frac{1}{2}$-averaged nonexpansiveness of $\mT_{\bG}$ and the $\frac{L}{2\lambda_{\min}(\bW)}$-averaged nonexpansiveness of $\mI-\bW^{-1}\nabla r$ (with respect to $\bW$) that $\mT_{\bW}$ is $\frac{2\lambda_{\min}(\bW)}{4\lambda_{\min}(\bW)-L}$-averaged nonexpansive with respect to $\bW$, which completes the proof.
\hspace*{\fill}~\QED

\section*{Acknowledgments}
This work was supported in part by the National Natural Science Foundation of China under Grants 12401120 and 62176103, in part by Guangdong Basic and Applied Basic Research Foundation under Grant 2021A1515110541, and in part by the Science and Technology Planning Project of Guangzhou under Grants 2024A04J3940 and 2024A04J9896.

\bibliographystyle{elsarticle-num}
\bibliography{bibfile}

\end{document}